\PassOptionsToPackage{table}{xcolor}
\documentclass[
  reprint,
  amsmath,
  amssymb,
  aps,
  pra
]{revtex4-2}

\usepackage{macros}
\usepackage{amsthm}
\usepackage{amsmath}
\usepackage{amssymb}
\usepackage{mathtools}
\usepackage{titlesec}
\usepackage{longtable}
\usepackage{rotating}
\allowdisplaybreaks %

\usepackage{multirow}  %

\newtheorem{theorem}{Theorem}[section]

\newtheorem{corollary}{Corollary}[theorem]

\numberwithin{theorem}{section} %
\numberwithin{corollary}{section}

\usepackage{xcolor}

\newcommand{\classicalexponent}{0.176 }
\newcommand{\classicalexponenterror}{0.011}
\newcommand{\classicalconstant}{19.369}
\newcommand{\classicalconstanterror}{0.657}

\newcommand{\classicalexponentwalksatlm}{0.345}
\newcommand{\classicalexponentdsatur}{0.42}

\DeclareMathOperator*{\argmin}{arg\,min}

\newcommand{\crossoverqubitsfaconespeeduptwo}{242}

\newcommand{\pfaconespeeduptwo}{71}

\newcommand{\yearsfaconespeeduptwo}{3}

\newcommand{\Tcountfaconespeeduptwo}{419.0}

\newcommand{\Tdepthfaconespeeduptwoa}{6750.25}

\newcommand{\Tfactoriesfaconespeeduptwo}{840}

\newcommand{\Tpqubitsfaconespeeduptwo}{152.43}

\newcommand{\distancefaconespeeduptwo}{35}

\newcommand{\Trotationsfaconespeeduptwo}{27}

\newcommand{\coresfaconespeeduptwo}{71}

\newcommand{\ancillafaconespeeduptwo}{43680}

\newcommand{\decodersfaconespeeduptwo}{52.32}

\newcommand{\rotfaconespeeduptwoa}{1.76}

\newcommand{\rotfaconespeeduptwob}{-09}

\newcommand{\Tfaconespeeduptwoa}{1.00}

\newcommand{\Tfaconespeeduptwob}{-17}

\newcommand{\oracleratiofaconespeeduptwo}{22}

\newcommand{\crossoverqubitsfaconespeedupthree}{191}

\newcommand{\pfaconespeedupthree}{253}

\newcommand{\hoursfaconespeedupthree}{64.57}

\newcommand{\Tcountfaconespeedupthree}{0.96}

\newcommand{\Tdepthfaconespeedupthreea}{20.76}

\newcommand{\Tfactoriesfaconespeedupthree}{592}

\newcommand{\Tpqubitsfaconespeedupthree}{84.43}

\newcommand{\distancefaconespeedupthree}{29}

\newcommand{\Trotationsfaconespeedupthree}{24}

\newcommand{\coresfaconespeedupthree}{50}

\newcommand{\ancillafaconespeedupthree}{30784}

\newcommand{\decodersfaconespeedupthree}{36.9}

\newcommand{\rotfaconespeedupthreea}{6.65}

\newcommand{\rotfaconespeedupthreeb}{-08}

\newcommand{\Tfaconespeedupthreea}{1.44}

\newcommand{\Tfaconespeedupthreeb}{-14}

\newcommand{\crossoverqubitsfaconespeedupfour}{179}

\newcommand{\pfaconespeedupfour}{623}

\newcommand{\hoursfaconespeedupfour}{14.99}

\newcommand{\Tcountfaconespeedupfour}{0.21}

\newcommand{\Tdepthfaconespeedupfoura}{5.0}

\newcommand{\Tfactoriesfaconespeedupfour}{540}

\newcommand{\Tpqubitsfaconespeedupfour}{73.91}

\newcommand{\distancefaconespeedupfour}{28}

\newcommand{\Trotationsfaconespeedupfour}{23}

\newcommand{\coresfaconespeedupfour}{46}

\newcommand{\ancillafaconespeedupfour}{28080}

\newcommand{\decodersfaconespeedupfour}{33.66}

\newcommand{\rotfaconespeedupfoura}{1.41}

\newcommand{\rotfaconespeedupfourb}{-07}

\newcommand{\Tfaconespeedupfoura}{6.48}

\newcommand{\Tfaconespeedupfourb}{-14}

\newcommand{\classicalpowerfaconespeedupfour}{269.27}

\newcommand{\oracleratiofaconespeedupfour}{23}

\newcommand{\pfacfivespeedupfour}{623}

\newcommand{\qubitshundredlogicalcyleone}{233}

\newcommand{\timehundredlogicalcyleone}{82.21}

\newcommand{\improvoptimisticcrossoverqubitsfactwospeedupfournumberzero}{185}

\newcommand{\improvoptimistichoursfactwospeedupfournumberzero}{7.23}

\newcommand{\improvoptimisticTpqubitsfactwospeedupfournumberzero}{40.84}

\newcommand{\improvoptimisticcrossoverqubitsfactwospeedupfournumberone}{195}

\newcommand{\improvoptimistichoursfactwospeedupfournumberone}{25.45}

\newcommand{\improvoptimisticTpqubitsfactwospeedupfournumberone}{25.77}

\newcommand{\improvoptimisticcrossoverqubitsfactwospeedupfournumbertwo}{203}

\newcommand{\improvoptimistichoursfactwospeedupfournumbertwo}{65.5}

\newcommand{\improvoptimisticTpqubitsfactwospeedupfournumbertwo}{29.82}

\newcommand{\improvoptimisticcrossoverqubitsfactwospeedupfournumberthree}{203}

\newcommand{\improvoptimistichoursfactwospeedupfournumberthree}{65.24}

\newcommand{\improvoptimisticTpqubitsfactwospeedupfournumberthree}{45.25}

\newcommand{\improvoptimisticcrossoverqubitsfactwospeedupfournumberfour}{177}

\newcommand{\improvoptimistichoursfactwospeedupfournumberfour}{2.94}

\newcommand{\improvoptimisticTpqubitsfactwospeedupfournumberfour}{8.88}

\newcommand{\improvpessimisticdaysfaconespeedupfournumberzero}{2.31}

\newcommand{\improvpessimisticTpqubitsfaconespeedupfournumberzero}{135.81}

\newcommand{\improvpessimisticdaysfaconespeedupfournumberone}{7.83}

\newcommand{\improvpessimisticTpqubitsfaconespeedupfournumberone}{77.31}

\newcommand{\improvpessimisticdaysfaconespeedupfournumbertwo}{20.74}

\newcommand{\improvpessimisticTpqubitsfaconespeedupfournumbertwo}{107.14}

\newcommand{\improvpessimisticdaysfaconespeedupfournumberthree}{20.7}

\newcommand{\improvpessimisticTpqubitsfaconespeedupfournumberthree}{153.67}

\newcommand{\improvpessimistichoursfaconespeedupfournumberfour}{21.75}

\newcommand{\improvpessimisticTpqubitsfaconespeedupfournumberfour}{29.92}

\newcommand{\improvpessimisticcrossoverqubitsfactwospeedupfournumberzero}{278}

\newcommand{\improvpessimisticcrossoverqubitsfactwospeedupfournumberone}{289}

\newcommand{\improvpessimisticcrossoverqubitsfactwospeedupfournumbertwo}{296}

\newcommand{\improvpessimisticcrossoverqubitsfactwospeedupfournumberthree}{296}

\newcommand{\improvpessimisticcrossoverqubitsfactwospeedupfournumberfour}{271}

\newcommand{\improvcrossoverqubitsfactwospeeduptwo}{235}

\newcommand{\improvdaysfactwospeeduptwo}{144.64}

\newcommand{\improvTcountfactwospeeduptwo}{265.24}

\newcommand{\improvTdepthfactwospeeduptwoa}{4594.56}

\newcommand{\improvTfactoriesfactwospeeduptwo}{400}

\newcommand{\improvTpqubitsfactwospeeduptwo}{17.0}

\newcommand{\improvdistancefactwospeeduptwo}{18}

\newcommand{\improvTrotationsfactwospeeduptwo}{27}

\newcommand{\improvcoresfactwospeeduptwo}{725760}

\newcommand{\improvancillafactwospeeduptwo}{20800}

\newcommand{\improvdecodersfactwospeeduptwo}{25.04}

\newcommand{\improvrotfactwospeeduptwoa}{1.76}

\newcommand{\improvrotfactwospeeduptwob}{-09}

\newcommand{\improvTfactwospeeduptwoa}{1.00}

\newcommand{\improvTfactwospeeduptwob}{-17}

\newcommand{\improvrotfaconespeedupthreea}{8.63}

\newcommand{\improvcrossoverqubitsfactwospeedupthree}{189}

\newcommand{\improvhoursfactwospeedupthree}{12.02}

\newcommand{\improvTcountfactwospeedupthree}{0.87}

\newcommand{\improvTdepthfactwospeedupthreea}{19.32}

\newcommand{\improvTfactoriesfactwospeedupthree}{296}

\newcommand{\improvTpqubitsfactwospeedupthree}{9.69}

\newcommand{\improvdistancefactwospeedupthree}{15}

\newcommand{\improvTrotationsfactwospeedupthree}{24}

\newcommand{\improvcoresfactwospeedupthree}{725760}

\newcommand{\improvancillafactwospeedupthree}{15392}

\newcommand{\improvdecodersfactwospeedupthree}{18.54}

\newcommand{\improvrotfactwospeedupthreeb}{-08}

\newcommand{\improvTfactwospeedupthreea}{1.57}

\newcommand{\improvTfactwospeedupthreeb}{-14}

\newcommand{\improvcrossoverqubitsfactwospeedupfour}{177}

\newcommand{\improvhoursfactwospeedupfour}{2.94}

\newcommand{\improvTcountfactwospeedupfour}{0.2}

\newcommand{\improvTdepthfactwospeedupfoura}{4.72}

\newcommand{\improvTfactoriesfactwospeedupfour}{270}

\newcommand{\improvTpqubitsfactwospeedupfour}{8.84}

\newcommand{\improvdistancefactwospeedupfour}{15}

\newcommand{\improvTrotationsfactwospeedupfour}{23}

\newcommand{\improvcoresfactwospeedupfour}{725760}

\newcommand{\improvancillafactwospeedupfour}{14040}

\newcommand{\improvdecodersfactwospeedupfour}{16.92}

\newcommand{\improvrotfactwospeedupfoura}{1.47}

\newcommand{\improvrotfactwospeedupfourb}{-07}

\newcommand{\improvTfactwospeedupfoura}{7.00}

\newcommand{\improvTfactwospeedupfourb}{-14}

\newcommand{\pessimisticcrossoverqubitsfaconespeeduptwo}{366}

\newcommand{\pessimisticyearsfaconespeeduptwo}{708}

\newcommand{\pessimisticTcountfaconespeeduptwo}{1241049.01}

\newcommand{\pessimisticTdepthfaconespeeduptwoa}{13296817.43}

\newcommand{\pessimisticTfactoriesfaconespeeduptwo}{1240}

\newcommand{\pessimisticTpqubitsfaconespeeduptwo}{72.22}

\newcommand{\pessimisticdistancefaconespeeduptwo}{22}

\newcommand{\pessimisticTrotationsfaconespeeduptwo}{27}

\newcommand{\pessimisticcoresfaconespeeduptwo}{725760}

\newcommand{\pessimisticancillafaconespeeduptwo}{64480}

\newcommand{\pessimisticdecodersfaconespeeduptwo}{77.25}

\newcommand{\pessimisticrotfaconespeeduptwoa}{1.76}

\newcommand{\pessimisticrotfaconespeeduptwob}{-09}

\newcommand{\pessimisticTfaconespeeduptwoa}{1.00}

\newcommand{\pessimisticTfaconespeeduptwob}{-17}

\newcommand{\pessimisticcrossoverqubitsfaconespeedupthree}{286}

\newcommand{\pessimistichoursfaconespeedupthree}{354.34}

\newcommand{\pessimisticTcountfaconespeedupthree}{73.22}

\newcommand{\pessimisticTdepthfaconespeedupthreea}{996.59}

\newcommand{\pessimisticTfactoriesfaconespeedupthree}{936}

\newcommand{\pessimisticTpqubitsfaconespeedupthree}{36.41}

\newcommand{\pessimisticdistancefaconespeedupthree}{17}

\newcommand{\pessimisticTrotationsfaconespeedupthree}{26}

\newcommand{\pessimisticcoresfaconespeedupthree}{725760}

\newcommand{\pessimisticancillafaconespeedupthree}{48672}

\newcommand{\pessimisticdecodersfaconespeedupthree}{58.32}

\newcommand{\pessimisticrotfaconespeedupthreea}{7.15}

\newcommand{\pessimisticTfaconespeedupthreea}{1.66}

\newcommand{\pessimisticTfaconespeedupthreeb}{-16}

\newcommand{\pessimisticrotfactwospeedupthreeb}{-09}

\newcommand{\pessimisticcrossoverqubitsfaconespeedupfour}{263}

\newcommand{\pessimistichoursfaconespeedupfour}{21.75}

\newcommand{\pessimisticTcountfaconespeedupfour}{4.31}

\newcommand{\pessimisticTdepthfaconespeedupfoura}{65.25}

\newcommand{\pessimisticTfactoriesfaconespeedupfour}{836}

\newcommand{\pessimisticTpqubitsfaconespeedupfour}{29.81}

\newcommand{\pessimisticdistancefaconespeedupfour}{16}

\newcommand{\pessimisticTrotationsfaconespeedupfour}{25}

\newcommand{\pessimisticcoresfaconespeedupfour}{725760}

\newcommand{\pessimisticancillafaconespeedupfour}{43472}

\newcommand{\pessimisticdecodersfaconespeedupfour}{52.1}

\newcommand{\pessimisticrotfaconespeedupfoura}{3.17}

\newcommand{\pessimisticrotfaconespeedupfourb}{-08}

\newcommand{\pessimisticTfaconespeedupfoura}{3.26}

\newcommand{\pessimisticTfaconespeedupfourb}{-15}

\begin{document}
\title{Threshold for Fault-tolerant Quantum Advantage\\ with the Quantum Approximate Optimization Algorithm}

\author{Sivaprasad Omanakuttan}
\email{sivparasad.thattupurackalomanakuttan@jpmchase.com}
\author{Zichang He}
\author{Zhiwei Zhang}
\author{Tianyi Hao}
\author{Arman Babakhani}
\author{Sami~Boulebnane}
\author{Shouvanik~Chakrabarti}
\author{Dylan~Herman}
\author{Joseph~Sullivan}
\author{Michael A.~Perlin}
\email{michael.perlin@jpmchase.com}
\author{Ruslan Shaydulin}
\email{ruslan.shaydulin@jpmchase.com}
\author{Marco Pistoia}
\affiliation{Global Technology Applied Research, JPMorganChase, New York, NY 10001, USA}

\begin{abstract}
  Optimization is often cited as a promising application of quantum computers.
  However, the low degree of provable quantum speedups has led prior rigorous end-to-end resource analyses to conclude that a quantum computer is unlikely to surpass classical state-of-the-art on optimization problems under realistic assumptions.
  In this work, we compile and analyze the Quantum Approximate Optimization Algorithm (QAOA) combined with Amplitude Amplification (AA) applied to random 8-SAT at the satisfiability threshold. 
  Our compilation involves careful optimization of circuits for Hamiltonian simulation, which may be of independent interest.
  We use the analytical scaling of the time-to-solution for QAOA identified by \citet{Sami_8_SAT_QAOA} and find that with QAOA depth $p=\pfacfivespeedupfour$, QAOA+AA achieves a crossover with state-of-the-art classical heuristics at \crossoverqubitsfaconespeedupfour{}
  variables and \hoursfaconespeedupfour{} hours of runtime when executed on a surface-code-based fault-tolerant quantum computer with \Tpqubitsfaconespeedupfour{} million physical qubits, a physical error rate of $10^{-3}$, and a $1~\mu$s code cycle time.
  Notably, we allow the classical solver to be parallelized as long as its total energy consumption is equal to that required for decoding in the surface code. 
  We further show that this restriction on classical solver energy consumption can be relaxed given optimistic but plausible reductions in physical error rates and fault-tolerance overheads, enabling a crossover of \improvoptimistichoursfactwospeedupfournumberfour{} hours using \improvoptimisticTpqubitsfactwospeedupfournumberfour{} million physical qubits against a classical solver running on a supercomputer with $725,760$ CPU cores.
  These findings support the hypothesis that large-scale fault-tolerant quantum computers will be useful for optimization.
\end{abstract}

\maketitle

\section{Introduction}

Optimization is often included in the list of the domains for which quantum computers are likely to have an impact~\cite{Abbas2024,Alexeev2021,Herman2023} due to the existence of many broadly applicable quantum algorithms with provable asymptotic speedups~\cite{quant-ph/9607014,montanaro2018quantum,montanaro2020quantum,2210.03210,Somma2008,Wocjan2008,Hastings2018,2212.01513,2410.23270}. However, most such speedups are obtained using a variant of amplitude amplification~\cite{quant-ph/9607014} and are therefore only quadratic~\cite{montanaro2018quantum,montanaro2020quantum,2210.03210,Somma2008,Wocjan2008}. Recently, algorithms with super-quadratic speedups have been proposed based on the short-path algorithm~\cite{Hastings2018,2212.01513,2410.23270}; however, their speedup is only slightly better than quadratic and is only over a restricted set of classical algorithms, namely brute force~\cite{Hastings2018,2212.01513} and Markov chain search~\cite{2410.23270}. In some restricted cases, quantum algorithms have been shown to achieve an exponential speedup over the best known %
classical algorithms~\cite{2411.04979,2408.08292,2503.12789}. However, it remains to be seen whether superpolynomial separations can be achieved for more general classes of optimization problems.

In addition to the small degree of the speedup, the practical applicability of these algorithms is limited by the high cost of their implementation, including the need for very deep circuits and extensive amounts of quantum arithmetic. Combined with overhead of error correction, these observations led prior resource analyses to conclude that quantum algorithms with small polynomial speedups in general~\cite{focus_beyond_quadratic} and for optimization in particular~\cite{Yuval_toffoli,Campbell2019,Dalzell2023,2412.13274} are unlikely to deliver a practical speedup on realistic large-scale fault-tolerant quantum computers. Notably, a detailed analysis of quantum algorithms for unstructured random $k$-SAT problems concluded that a speedup is unlikely once the classical cost of decoding is taken into account~\cite{Campbell2019}. For structured problems, the speedup is absent even if the cost of decoding is ignored~\cite{2412.13274}.

The quantum approximate optimization algorithm (QAOA)~\cite{Hogg2000,Hogg2000search,farhi2014quantumapproximateoptimizationalgorithm} is a quantum heuristic for optimization that has recently been shown to provide a polynomial speedup over classical state-of-the-art for the Low Autocorrelation Binary Sequences~\cite{Shaydulin2024} and random 8-SAT~\cite{Sami_8_SAT_QAOA} problems. QAOA solves optimization problems by using a parameterized quantum circuit consisting of $p$ steps of applying, in alternation, two Hamiltonian evolution operators.
The first operator is called the ``phaser'' and applies a phase to computational basis states proportionally to the objective function value of the corresponding bitstring.
The second operator is called the ``mixer'' and induces non-trivial dynamics equivalent to a quantum walk on the boolean hypercube.
A notable attribute of QAOA is the simplicity of its circuit, consisting of repetitions of two fast-forwardable Hamiltonian evolutions, which has enabled small-scale hardware demonstrations~\cite{Shaydulin2023npgeq,Pelofske2023,Pelofske2024,2409.12104,Tasseff2024}.

In this work, we focus on the random 8-SAT problem, for which \citet{Sami_8_SAT_QAOA} analytically derive the instance-average success probability and show that it decays exponentially with problem size $n$, with empirical results showing a power-law dependency on QAOA depth $p$. Specifically, the time-to-solution of QAOA with $p$ layers for 8-SAT is shown to be $O(2^{0.69p^{-0.32}n})$.
The exponent can be reduced by another factor of two by using amplitude amplification (QAOA+AA).
For sufficiently large $p$, this approach gives a polynomial speedup over state-of-the-art classical algorithms.

Our main result is an analysis of the conditions for a quantum advantage for the random 8-SAT problem using QAOA+AA compiled to a fault-tolerant quantum processor based on the surface code.
On the quantum algorithm side, we compile QAOA+AA and optimize multiple aspects of the circuit, including reducing the circuit depth required to implement QAOA phaser and $k$-SAT oracle, and identifying the optimal number of $T$ gates to maximize the fidelity of Hamiltonian evolutions, among other improvements. As some of these optimizations apply to Hamiltonian evolution more broadly and are not specific to QAOA, they may be of independent interest. We note that our logical circuit is composed only of single-qubit Pauli gates, single- and two-qubit Pauli measurements, and the preparation of logical $\ket{0}$, $\ket{+}$, $T$, and CCZ states, allowing us to give a complete resource estimate.

On the classical algorithm side, we benchmark state-of-the-art classical heuristics for random 8-SAT and identify Sparrow \cite{balint2014engineering} as the most performant one, with a time-to-solution that scales as $2^{\classicalexponent n}$ with the number of variables $n$, which is a notable improvement over the scaling of $2^{\classicalexponentdsatur{}n}$ and $2^{\classicalexponentwalksatlm{}n}$ used as the classical point of comparison in prior analyses of the prospects for a quantum advantage for random 8-SAT~\cite{Campbell2019,Sami_8_SAT_QAOA}. %
We further allow for the parallelization of Sparrow, whose effective speedup from running on multiple cores was studied in Ref.~\cite{arbelaez2013using}, and assume that Sparrow parallelizes as well as the most parallelizable unstructured problem instance considered in Ref.~\cite{arbelaez2013using}.
The number of CPUs used by Sparrow in our analysis is determined by energy consumption; namely, we require that the total energy consumption of the CPUs running Sparrow is equal to that required to perform real-time decoding of the surface code \cite{Barber2025}.

We find that for random 8-SAT on $n=\crossoverqubitsfaconespeedupfour{}$ variables, QAOA+AA with $p=\pfaconespeedupfour{}$ executed on $\Tpqubitsfaconespeedupfour{}\times 10^6$ physical qubits has an expected time-to-solution of \hoursfaconespeedupfour{} hours, which is equal to that of state-of-the-art classical solver running on \coresfaconespeedupfour{} CPU cores with a power budget of \classicalpowerfaconespeedupfour{} W.
For $n=\qubitshundredlogicalcyleone{}$ variables, QAOA+AA achieves $100\times$ speedup over classical state-of-the-art with a quantum runtime of \timehundredlogicalcyleone{} hours.

Our main result relies only on existing techniques and a relatively conservative set of hardware assumptions. However, we anticipate that the rapid progress in error correction and hardware will lead to further reductions in resource requirements, which are not taken into account in our main results.
We therefore analyze the impact of potential future progress in magic state cultivation~\cite{gidney2024magic}, algorithmic fault-tolerance~\cite{zhou2024algorithmic}, and hardware error rates.
We find that these improvements may lead to a crossover time of only \improvoptimistichoursfactwospeedupfournumberfour{} hours using \improvoptimisticTpqubitsfactwospeedupfournumberfour{}$\times 10^6$ physical qubits against classical solver running on all $725,760$ cores of MareNostrum 5 GPP supercomputer~\cite{top500}.

Our results contrast with the negative findings of previous works~\cite{focus_beyond_quadratic,Yuval_toffoli,Campbell2019} and are enabled by combining multiple contributions with recent developments in the field.
First, we leverage improved decoding techniques~\cite{Barber2025,Wu2023,Higgott2025}, which have led to a $100\times$ reduction in the classical computing overheads of decoding as compared to the estimates used in Ref.~\cite{Campbell2019}.
Second, we leverage improved resource-state factories \cite{litinski2019magic} and (similarly to Ref.~\cite{Campbell2019}) allow for their parallelization, such that the runtime of the quantum algorithm is not limited by the time required to produce a resource state.
Third, we trade space for time by optimizing the compilation, scheduling, and parallelism of circuit components.
These techniques reduce the effective time cost of one clause in the QAOA phaser to one logical QEC cycle ($d$ rounds of syndrome extraction in a distance-$d$ surface code).
Finally, we use QAOA+AA, which has a simple circuit and, as shown by \citet{Sami_8_SAT_QAOA}, offers a super-quadratic speedup.

\section{Results}

\begin{figure*}
  \centering
  \includegraphics{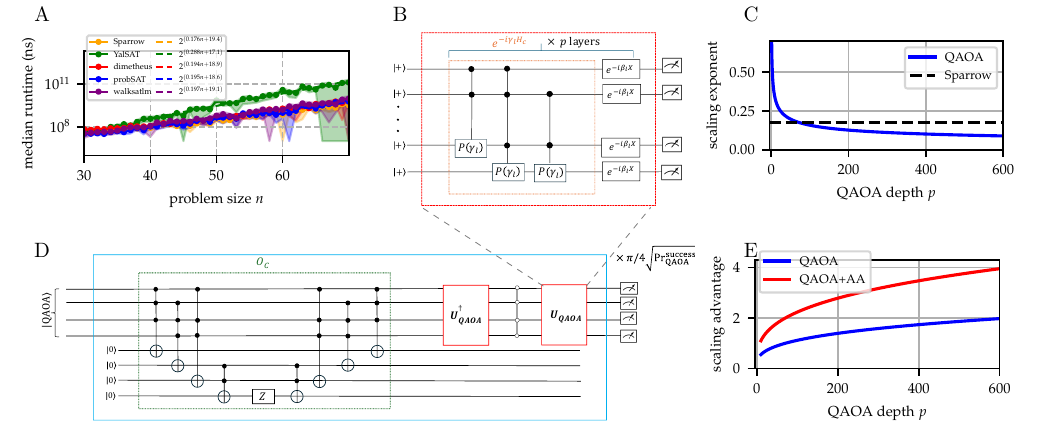}
  \caption{
    \textbf{Classical and quantum algorithms for 8-SAT.}
    \textbf{(A)} Median time-to-solution (TTS) of state-of-the-art classical solvers for random 8-SAT near the satisfiability threshold with a clause-to-variable ratio of $r=176$.
    The best scaling is achieved by the Sparrow solver, for which the TTS is $2^{\classicalexponent n + \classicalconstant} \, \mathrm{ns}$.
    \textbf{(B)} Example QAOA circuit for an instance of 3-SAT, which is composed of $p$ layers that alternate the phaser and mixer, where the phaser (orange dotted box) encodes the $k$-SAT objective function and the mixer consists of single-qubit rotations $R_{\rm x}(\beta_l) = e^{-i\beta_l X/2}$.
    \textbf{(C)} The TTS scaling exponent of QAOA with varying QAOA depths $p$, reproduced from Ref.~\cite{Sami_8_SAT_QAOA}.
    The dashed line indicates the scaling exponent for Sparrow, for reference.
    \textbf{(D)} The full circuit for QAOA+AA consists of $\pi/(4\sqrt{\Pr_{\rm QAOA}^{\rm success}})$ repetitions of a module that contains two QAOA circuits, an oracle $O_C$ that flips the sign of solutions to $k$-SAT, and a ``zero-state oracle'' that flips the sign of the all-0 state.
    \textbf{(E)} The asymptotic speedup of QAOA and QAOA+AA over the state-of-the-art classical solver Sparrow for varying values of QAOA depth $p$.
    QAOA+AA scales quadratically better than QAOA, enabling larger speedups.
  }
  \label{fig:algorithms}
\end{figure*}

We study the regime of quantum advantage for the random $k$-SAT problem~\cite{CojaOghlan2009}, which has been studied extensively from the perspective of both classical \cite{dequen2006efficient,marino2016backtracking,coja2010better,aguirre2001random} and quantum~\cite{Campbell2019,Campos2021,zhang2022quantum,Sami_8_SAT_QAOA,2412.13274,2411.17442} algorithms.
In this problem, the objective is to determine whether there exists an assignment of $n$ truth values to variables that satisfies a Boolean formula with $m$ clauses, where each clause in the formula contains exactly $ k $ literals.
We defer formal definition of the ensemble of instances used to Methods.
As the clause-to-variable ratio $\frac{m}{n}$ grows, $k$-SAT undergoes a phase transition from random instances being satisfiable with high probability to being almost surely unsatisfiable at $\frac{m}{n} \approx 2^k \log{2}$ (up to the leading order in $k$)~\cite{CojaOghlan2009}. %
Whereas efficient classical algorithms or performant heuristics exist for $k$-SAT with $\frac{m}{n} < 2^k\frac{\log{k}}{k}$~\cite{coja2010better},
it remains an open question to understand the power of algorithms to address problems with $2^k\frac{\log{k}}{k} < \frac{m}{n} \lesssim 2^k \log{2}$.
Following Ref.~\cite{Sami_8_SAT_QAOA}, we focus on random 8-SAT close to the satisfiability threshold and set $\frac{m}{n}=176$.

\subsection{Clasical solvers}

We begin by benchmarking the state-of-the-art classical solvers.
We include in our comparison incomplete classical algorithms, i.e., algorithms that cannot deduce that an instance is unsatisfiable.
We remark that our quantum algorithm, QAOA, is also an incomplete solver.
We include winners of the three latest random track SAT competitions (Sparrow \cite{balint2014engineering}, 2018 \cite{heule2019sat}, yalSAT \cite{biere2016splatz}, 2017 \cite{balyosat} and Dimetheus \cite{gableske2013performance}, 2016 \cite{balyo2017sat}; random track removed in 2019 \cite{heule2019benchmark}).
Additionally, we include probSAT \cite{balint2012choosing} and WalkSATlm \cite{cai2013improving}, which are known to perform well on random problems \cite{fu2022improving}.
Fig.~\ref{fig:algorithms}A shows the growth of time-to-solution (TTS) with number of variables, with additional details deferred to Methods.
We find Sparrow \cite{balint2014engineering} to be the most performant solver, with median TTS given by
\begin{align}
  T_c\left(n\right) = 2^{\left(\classicalexponent{}\pm\classicalexponenterror{}\right) n+\left(\classicalconstant{}\pm\classicalconstanterror{}\right)}~ \mathrm{ns},
\end{align}
where the error bars denote a $90\%$ confidence interval.
This scaling is a notable improvement over the scaling exponents \classicalexponentdsatur{} and \classicalexponentwalksatlm{} used in prior analyses of the prospects for a quantum advantage for random $k$-SAT~\cite{Campbell2019,Sami_8_SAT_QAOA}.
We remark that our benchmarking finds a lower exponent of $0.197$ for WalkSATlm than that reported Ref.~\cite{Sami_8_SAT_QAOA} due to our use of an optimized implementation \cite{cai2013improving, walksatlm_implementation}.

The parallelization behavior of Sparrow, and specifically the effective speedup obtained by running on multiple cores, was studied in Ref.~\cite{arbelaez2013using}.
Therein, the authors found that the effective speedup from parallelization is linear in the number of cores for structured boolean satisfiability problems, but sublinear for unstructured problem instances.
We allow for the parallelization of Sparrow in our analysis, and assume that Sparrow parallelizes as well as the most parallelizable unstructured problem instance (Rand-9) considered in Ref.~\cite{arbelaez2013using}.
The number of CPUs used by Sparrow is determined by setting their power budget ($\frac{280}{48}\approx5.8$ watts per CPU; see \cref{sec:classical_solvers}) equal to that required to perform real-time decoding of surface code (8 milliwatts per decoder) \cite{Barber2025}. We remark that the parallelization of Sparrow in the regime corresponding to our main results is near-perfect since the number of CPU cores is modest ($\leq\coresfaconespeeduptwo$).
See \cref{sec:number_of_decoders,sec:parallel_classical_solvers} for additional details.

\subsection{Quantum algorithm}

We consider QAOA, which solves satisfiability problems using a parameterized quantum state (shown in \cref{fig:algorithms}B),
\begin{align}
  \ket{\rm QAOA} &= U_{\rm QAOA} \ket{+}^{\otimes n}, \\
  U_{\rm QAOA} &= \prod_{l=1}^{p} e^{-i\beta_l H_M} e^{-i\gamma_l H_C},
  \label{eq:qaoa_parametrized_state}
\end{align}
where $\beta$ and $\gamma$ are free parameters; $p$ is the number of alternating layers, also called the QAOA depth; $\ket{+}^{\otimes n} \propto \sum_{x\in\set{0,1}^n} \ket{x}$ is a superposition over all computational basis states; $H_C$ is the cost Hamiltonian encoding the optimization problem; and we set the mixer Hamiltonian to $H_M = \frac12\sum_{j=1}^n X_j$ with $X_j = \op{0}{1}_j + \op{1}{0}_j$.

The cost Hamiltonian $H_C$ for $k$-SAT energetically rewards satisfying clauses,
\begin{equation}\label{eq:QAOA_cost_Hamiltonian}
  H_C = -\sum_{j=1}^m  \sum_{x \in \set{0,1}^n} C_j(x)\op{x},
\end{equation}
where each clause can be written in the form
\begin{equation}
  C_j(x) = \ell_{j1}(x_{j1})\lor \ell_{j2}(x_{j2})\lor \cdots \lor \ell_{jk}(x_{jk}).
  \label{eq:one_clause}
\end{equation}
Here $x_{ji}$ is the $i$-th variable addressed by clause $C_j$, and $\ell_{ji}(x_{ji})=x_{ji}$ or $\ell_{ji}(x_{ji})=\neg x_{ji} = 1 - x_{ji}$ depending on the clause.
Note that this notation is over-complete, as the same variable may be addressed by multiple clauses.

Ref.~\cite{Sami_8_SAT_QAOA} derives an analytical expression for instance-average success probability of QAOA with $p$ layers. Evaluating this expression with optimized parameters $\beta$, $\gamma$, Ref.~\cite{Sami_8_SAT_QAOA} obtains $\Pr^{\rm success}_{\rm QAOA} = 2^{-0.69p^{-0.32}n}$.
The expected TTS of QAOA is equal to the inverse of the success probability (up to a polynomial cost for QAOA circuit implementation).
As Fig.~\ref{fig:algorithms}C shows, QAOA with a moderate depth achieves a speedup over Sparrow.

The scaling of TTS for QAOA with problem size $n$ can be improved quadratically by boosting the success probability with amplitude amplification (AA)~\cite{brassard2000quantum}.
The combined QAOA+AA circuit, summarized in Fig.~\ref{fig:algorithms}D, applies the following operator $\frac{\pi}{4\sqrt{\Pr^{\rm success}_{\rm QAOA}}}$ times~\cite{brassard2000quantum}:
\begin{align}
  Q = U_{\rm QAOA} O_0 U_{\rm QAOA}^\dag O_C.
  \label{eq:QAOA+AA}
\end{align}
Here $O_0$ is a ``zero-state oracle'' that flips the sign of the all-0 state, $O_0\ket{x}=-\ket{x}$ if $x=(0,\cdots,0)$ and $O_0\ket{x} = \ket{x}$ otherwise; and the $O_C$ is a $k$-SAT oracle defined by $O_C\ket{x}=(-1)^{C(x)}\ket{x}$, where $C(x) = 1$ for solutions $x$ to the 8-SAT problem instance and $C(x)=0$ otherwise.
Leveraging the quadratic speedup from AA, QAOA+AA improves the time-to-solution for random $8$-SAT to $2^{0.345p^{-0.32}n}$.
The asymptotic speedup of QAOA and QAOA+AA over the state-of-the-art Sparrow solver as a function of QAOA depth $p$ is shown in \cref{fig:algorithms}E.

We note that we have reported an instance-averaged TTS for both the classical (Sparrow) and quantum (QAOA+AA) algorithms.
However, for a given instance, the required time to find a solution could be higher or lower than the average.
When running a classical algorithm in practice, we can estimate an upper bound on the required runtime and implement a check that stops the run early whenever a satisfying assignment is found.
In the quantum case, additional care is required to avoid ``overshooting'' the target state with amplitude amplificaton.
We show in Theorem~\ref{th:quantum_overshooting} of Methods how to prevent overshooting at the cost of quadrupling the quantum runtime, and include this overhead in our analysis.
For clarity of presentation, we use the instance-averaged TTS as the runtime for both classical and quantum algorithms; however, the crossover point does not change if both times are scaled by a constant.

\begin{figure*}
  {\includegraphics{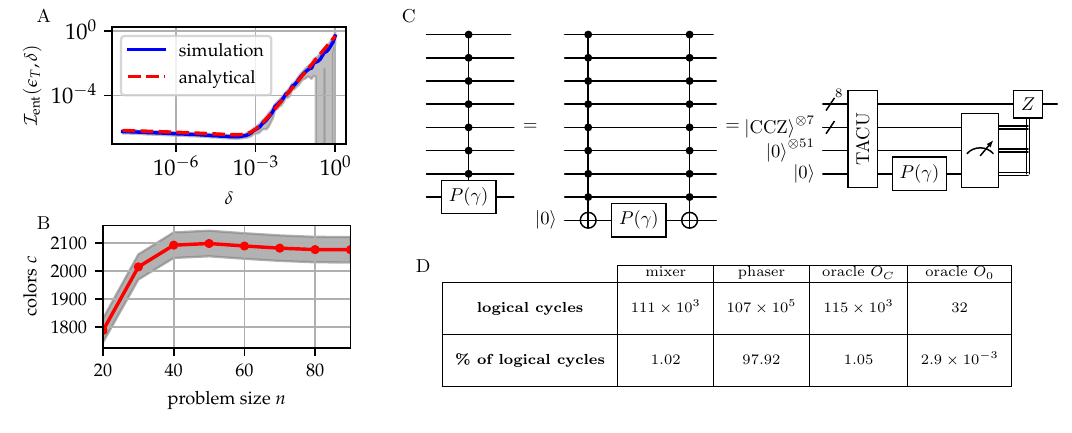}}
  \caption{\textbf{Compilation of QAOA+AA for 8-SAT.}
    \textbf{(A)} Given a $T$-state infidelity $\epsilon_T$ (here $\epsilon_T=10^{-8}$), we optimize the accuracy $\delta$ with which phase gates $P(\gamma) = e^{i\gamma\op{1}}$ are decomposed into a $\mathrm{Clifford}+T$ gate set by minimizing the entanglement infidelity $\mathcal{I}_{\mathrm{ent}}$ of the decomposition (see \cref{sec:optimal_decomp_arb_rotation} of the Methods).
    Increasing $\delta$ leads to greater decomposition error, while smaller $\delta$ increases the number of noisy $T$ gates in the decomposition.
    The dashed line shows analytical estimates for $\mathcal{I}_{\mathrm{ent}}$, while the solid line and shaded region shows the mean and standard deviation from numerical trials with $40$ random angles.
    \textbf{(B)} The number of colors $c$ empirically found to be necessary to color the clause collision graph for random instances of 8-SAT with different problem sizes $n$ and a clause-to-variable ratio of $r=\frac{m}{n}=176$.
    These colors are used to partition 8-SAT clauses into $c$ subsets such that every subset consists of clauses that address mutually disjoint sets of qubits.
    The value of $c$ sets a lower bound on the runtime of the QAOA phaser for $k$-SAT.
    \textbf{(C)} Schematic of the circuit to implement an 8-qubit phase gate $P_8(\gamma) = e^{i\gamma\op{1}^{\otimes 8}}$, which requires 7 CCZ states and 52 logical $\ket{0}$-state ancillas.
    The operation schematically labelled ``TACU'' addresses the top 8 qubits for only one logical cycle, and the final operation labelled ``$Z$'' indicates the application of Pauli-$Z$ correction gates as determined by measurement outcomes.
    These correction gates can be deferred to the end of the phaser, allowing a new batch of clauses to be dispatched at every logical cycle of the phaser (see \cref{sec:phase_op_resources} of the Methods and \cref{sec:TACU}).
    \textbf{(D)} The time budget for different components of the QAOA+AA circuit corresponding to the crossover point for a cubic speedup in \cref{fig:main_results}.
    The $k$-SAT phaser dominates this budget, highlighting the importance of its optimization.
  }
  \label{fig:techniques}
\end{figure*}

\subsection{Quantum runtime}

To compute the quantum TTS, we analyze the cost of implementing the QAOA+AA circuit in detail.
The quantum TTS $T_q$ is equal to the combined runtime of the four components in Eq.~\eqref{eq:QAOA+AA}, also shown in \cref{fig:algorithms}D:
\begin{align}
  T_q
  = \frac{\pi}{4\sqrt{\Pr^{\rm success}_{\rm QAOA}}}
  \left[2 p\left(T_{\mathrm{mixer}}+T_{\mathrm{phaser}}\right)+T_{O_C}+T_{O_0}\right].
  \label{eq:each_component}
\end{align}
Here $T_{\mathrm{mixer}}$, $T_{\mathrm{phaser}}$, $T_{O_C}$, and $T_{O_0}$ are, respectively, the runtimes of the QAOA mixer $e^{-i\beta H_M}$, the QAOA phaser $e^{-i\gamma H_C}$, the $k$-SAT oracle $O_C$, and the zero-state oracle $O_0$.

To evaluate the runtime of these components, we consider their fault-tolerant implementation in a surface code architecture with a $\mathrm{Clifford} + T + \mathrm{CCZ}$ gate set that is realized with lattice surgery \cite{horsman2012surface} and resource-state factories \cite{litinski2019magic}.
Our runtimes will be provided in units of the logical cycle time $T_{\mathrm{LC}} = d\times 1~\mu$s, or the time needed to perform $d$ rounds of syndrome measurement in a distance-$d$ surface code.
For reference, we find that code distances $d\approx30$ are sufficient for all crossovers reported in this work (see \cref{sec:distance_surface_code}), in agreement with findings elsewhere in the literature \cite{focus_beyond_quadratic,Gidney2021}.

The mixer acts independently on each qubit, and can thereby be reduced to $n$ parallel implementations of the single-qubit phase gate $P(\gamma) = e^{i\gamma\op{1}}$ (in the $X$ basis).
We consider the approximation of single-qubit phase gates in the $\textrm{Clifford}+T$ gate set with the ancilla-assisted mixed fallback method of Ref.~\cite{kliuchnikov2023shorter}, and identify the required decomposition accuracy for achieving a QAOA+AA circuit fidelity of 99\% with a depolarizing noise model (see \cref{fig:techniques}A and \cref{sec:optimal_decomp_arb_rotation} of the Methods).
Denoting the number of logical cycles required to thus implement $P(\gamma)$ to the desired accuracy by $n_P$, which is approximately equal to the number of $T$ gates of the decomposition, we use the method of Ref.~\cite{litinski2022active} to consume one $T$ state per logical cycle, and account for a sufficient number of $T$-state factories to match this $T$ state consumption rate.
The runtime of the mixer is then
\begin{align}
  T_{\mathrm{mixer}} = n_P \times T_{\mathrm{LC}}.
\end{align}
While the precise value of $n_P$ depends on the choice of QAOA depth $p$, for reference we note that $n_P\in[25, 29]$ for the cases analyzed in this work.

The phaser $e^{-i\gamma H_C}$ applies, for each clause $C_j$ of the form in Eq.~\eqref{eq:one_clause}, a phase $e^{i\gamma}$ to all bitstrings that satisfy the clause.
To minimize the runtime of the phaser, we use a graph coloring algorithm \cite{osti_960616} to partition clauses into subsets such that every subset consists of clauses that address mutually disjoint sets of qubits.
Specifically, construst a \textit{clause collision graph}, which associates each clause with the node of a graph, and draw an edge between any pair of clauses that address the same qubit.
Coloring the nodes in this graph thereby identifies, by color, subsets of clauses that can be applied in parallel.
A simple counting argument suggests that a $k$-SAT instance with a clause-to-variable ratio of $r=\frac{m}{n}$ should be colorable with $c\sim k^2 r$ colors: each clause has $k$ variables, each of which belong to $\frac{km}{n} = kr$ clauses on average.
The graph we construct can be thereby expected to have degree $\sim k^2 r$, and by Vizing's theorem \cite{vizing1964estimate} be colorable with $\sim k^2 r$ colors.
As we show in \cref{fig:techniques}B, for $k=8$, a clause-to-variable ratio of $r=\frac{m}{n}=176$, and $n\gtrsim40$, this argument empirically overestimates the number of colors by a factor of $\sim 5$, and that in practice clauses can be partitioned into $c\approx12r$ subsets of average size $\frac{m}{c}\approx\frac{n}{12}$.

\begin{figure*}
  \begin{tabular}{|c|c|c|c|c|c|}
    \multicolumn{6}{l}{A} \\ \hline
    \textbf{asymptotic speedup}
    & \textbf{QAOA depth} $p$
    & \textbf{problem size} $n$
    & \textbf{code distance} $d$
    & \textbf{physical qubits} ($\times 10^6$)
    & \textbf{crossover time} \\
    \hline\hline
    quadratic
    & \pfaconespeeduptwo
    & \crossoverqubitsfaconespeeduptwo
    & \distancefaconespeeduptwo
    & \Tpqubitsfaconespeeduptwo
    & \yearsfaconespeeduptwo{} y
    \\
    \hline
    cubic
    & \pfaconespeedupthree
    & \crossoverqubitsfaconespeedupthree
    & \distancefaconespeedupthree
    & \Tpqubitsfaconespeedupthree
    & \hoursfaconespeedupthree{} h
    \\ \hline
    quartic
    & \pfaconespeedupfour
    & \crossoverqubitsfaconespeedupfour
    & \distancefaconespeedupfour
    & \Tpqubitsfaconespeedupfour
    & \hoursfaconespeedupfour{} h
    \\ \hline
  \end{tabular}
  \\
  \includegraphics{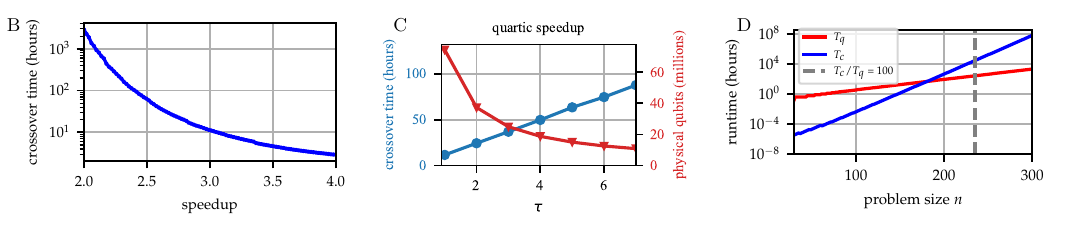}
  \caption{\textbf{Crossover times for random 8-SAT near the satisfiability threshold.}
    \textbf{(A)} Parameters for crossover points at which QAOA+AA and the classical solver Sparrow take, in expectation, equal time to solve random instances of 8-SAT near the satisfiability threshold.
    Crossover points are shown for QAOA depths $p$ that correspond to asymptotically quadratic, cubic, and quartic speedups.
    Here $n$ is the number of 8-SAT variables at the crossover point and $d$ is the surface code distance required to achieve a QAOA+AA circuit fidelity of $99\%$.
    We also report the number of physical qubits required when accounting for both logical ancilla qubits and resource-state factories from Ref.~\cite{litinski2019magic}, as well as the runtime of the algorithms at the crossover point.
    \textbf{(B)} The crossover time as a function of the asymptotic speedup for QAOA+AA over Sparrow.  The speedup is determined by the QAOA depth $p$, and the QAOA+AA algorithm has a crossover time of less than a
    day for a large range of speedups.
    \textbf{(C)} Different choices of a spacetime tradeoff in the QAOA+AA algorithm can reduce qubit overheads at a cost of increasing the crossover time.
    Specifically, we show how the time and space requirements for the crossover point with an asymptotically quartic speedup vary with respect to a ``slowdown factor'' $\tau$ that controls the number of logical cycles between dispatched clauses in the $k$-SAT phaser.
    Dispatching clauses at a slower rate (increasing $\tau$) reduces the resource-state consumption rate of the phaser, thereby requiring fewer resource-state factories and allowing for more reuse of ancilla qubits that are devoted to applying phases.
    \textbf{(D)} Classical runtimes $T_c$ for Sparrow and quantum runtime $T_q$ for QAOA+AA as a function of the problem size $n$.
    Here the QAOA depth $p$ is chosen for an asymptotically quartic speedup.
    The dashed vertical line indicates the problem size, $n=\qubitshundredlogicalcyleone$, for which $T_c/T_q = 100$.
  }
  \label{fig:main_results}
\end{figure*}

We now consider the time required to apply the phase for each clause.
The phase for each clause can be implemented by a $k$-qubit phase gate $e^{i\gamma \op{1}^{\otimes k}}$ that is sandwiched by bit-flip (Pauli-$X$) gates on variables that are negated in the clause.
In turn, the $k$-qubit phase gate can nominally be implemented by sandwiching a single-qubit phase gate $P(\gamma)$ between two $(k+1)$-qubit Toffoli gates.
In practice, these two Toffoli gates can be merged into a single $k$-qubit temporary-AND-compute-and-uncompute (TACU) gadget that ``dispatches'' the phase gate to an ancilla qubit, measures that ancilla qubit, and adaptively applies Pauli-$Z$ corrections to the $k$ variable qubits, shown schematically in \cref{fig:techniques}C.
We optimize the two-qubit TACU gadget of Ref.~\cite{litinski2022active} to construct a $k$-qubit TACU gadget that consumes $k-1$ CCZ states, addresses the $k$-SAT variable qubits for one logical cycle, and can be executed in a total of $n_P + 4\lceil\log_2 k\rceil$ logical cycles (see \cref{sec:TACU}).
This gadget allows us to dispatch one batch of clauses at each logical cycle of the $k$-SAT phaser, deferring Pauli-$Z$ corrections to the end of the phaser.
In total, the phaser can be implemented by dispatching clauses for $c$ logical cycles, waiting for all gadgets to complete, and applying adaptive Pauli-$Z$ corrections, such that the total runtime of the phaser is
\begin{align}
  T_{\mathrm{phaser}} = (c  + n_P + 4\lceil\log_2(k)\rceil) \times T_{\mathrm{LC}}.
\end{align}
We provide additional details in \cref{sec:phase_op_resources} and \cref{sec:TACU}.

The $k$-SAT oracle $O_C$ flips the sign of states that satisfy $m$ clauses, each of which is an OR of $k$ bits.
Similarly to the phaser, we use a $k$-qubit TACU gadget to flag the satisfaction of each clause with an ancilla qubit, and use a graph coloring algorithm to identify clauses whose satisfaction can be computed in parallel.
In turn, the ancilla qubits that flag the satisfaction of individual clauses can be pairwise AND-ed in a binary tree to flag the satisfaction of all clauses with one qubit, at which point the oracle phase is applied with a single-qubit Pauli-$Z$ gate.
A time-optimal implementation of the $k$-SAT oracle leads to an increase in both the CCZ consumption rate and the number of computational ancilla qubits required by the QAOA+AA algorithm, thereby greatly increasing space overheads.
Our implementation of the oracle therefore includes intentional delay times to reduce its CCZ consumption rate and allow for the recycling of ancilla qubits that compute the satisfaction of clauses.
Altogether, up to minor corrections (see  \cref{sec:oracle-bounds}) we find that the oracle can be implemented with runtime
\begin{align}
  T_{O_C} = 4c\log_2\left(\frac{km}{c}\right) \times T_{\mathrm{LC}}
\end{align}
without exceeding the qubit requirements of the phaser.

The zero-state oracle $O_0$ is equal, up to conjugation by Pauli-$X$ gates, to an $n$-qubit multi-control $\pi$-phase ($Z$) gate.
This gate can be implemented with a multi-qubit TACU gadget that consumes $n-1$ CCZ states.
The total runtime of the zero-state oracle is then
\begin{align}
  T_{O_0} = 4 \log_2 n \times T_{\mathrm{LC}}.
\end{align}
As illustrated by the time budget in \cref{fig:techniques}D, the runtime $T_q$ is dominated by the time to implement the QAOA phaser.

\subsection{Qubit requirements}

The qubit requirements of the QAOA+AA algorithm are dominated by two contributions: the number of logical ancilla qubits required for TACU gadgets, and the number of resource-state factories required to match $T$ and CCZ consumption rates of the algorithm.
Both contributions are, in turn, determined by the number of TACU gadgets that run in parallel as they are dispatched in the phaser.
If $s$ TACU gadgets dispatched at every logical cycle and each gadget runs for $\lambda$ logical cycles, the qubits used for the first batch of gadgets can be used for the batch at logical cycle $\lambda+1$, such that the total number of TACU gadgets that run in parallel is
\begin{align}
  n_{\mathrm{jobs}} = s\times \lambda.
  \label{eq:jobs}
\end{align}
For the cases analyzed in this work, $s \approx \frac{m}{c} \approx \frac{n}{12} \lesssim 20$ and $\lambda = 4\lceil\log_2(k)\rceil + n_P \le 41$.

Every TACU gadget in the phaser uses $\lfloor\frac{13}{2}k\rfloor$ logical ancilla qubits (see \cref{sec:TACU}).
For resource-state factories, we use the (15-to-1)$^4_{13,5,5}\times$(20-to-4)$_{27,13,15}$ $T$ factory and (15-to-1)$^6_{13,7,7}\times$(8-to-CCZ)$_{25,15,15}$ CCZ factory developed in Ref.~\cite{litinski2019magic}, and account for a sufficient number of factories to produce one $T$ state per TACU gadget per logical cycle.
The number of qubits required for this $T$-state production rate is sufficient to produce CCZ states at the rate that is required in CCZ consumption stages of the QAOA+AA algorithm, so we repurpose $T$ state factories into CCZ state factories on an as-needed basis.
We note that each of our $T$ state factories take just over 5 logical cycles to produce 4 $T$ states, so we need about 1.3 factories per TACU gadget.
In principle, one could slow down the $T$ state consumption stages of the QAOA+AA algorithm to consume $T$ states at a rate that is commensurate with the production rate of one factory.
While simplifying architectural (e.g., qubit layout and routing) considerations for the QAOA+AA algorithm, we note that the optimization of resource-state factories is an active area of research with notable recent advancements \cite{gidney2024magic} and pathways for further improvement that have not been accounted for in this work.
We therefore defer a more detailed analysis of architectural considerations with improved resource-state factories to future work.

\subsection{Quantum-classical crossover}

Fig.~\ref{fig:main_results}A provides a summary of quantum resources at the crossover point at which $T_q = T_c$ for choices of the QAOA depth $p$ that correspond to asymptotically quadratic, cubic, and quartic speedups over state-of-the-art classical solvers.
Even when accounting for the benefits of classical parallelization, our analysis finds, for example, that QAOA+AA is capable of outperforming state-of-the-art classical solvers on 8-SAT instances with $n = \crossoverqubitsfaconespeedupfour{}$ with a few hours of runtime using \Tpqubitsfaconespeedupfour{} million physical qubits.
A more detailed dependence of the crossover time on the asymptotic speedup is provided in
\cref{fig:main_results}B.
We note that the parameters for the crossover point depend on a choice of spacetime tradeoffs in the quantum algorithm, and show in \cref{fig:main_results}C how the crossover time and qubit overheads change if the QAOA phaser is slowed down by a factor of $\tau$ to decrease the rate at which the phaser dispatches TACU gadgets.
A lower dispatch rate reduces the number of parallel jobs $n_{\mathrm{jobs}}$ that run during the phaser, thereby reducing both the number of ancillas that are required for these jobs, and the number of resource-state factories required to maintain the resource-state consumption rate of the phaser. Finally, in \cref{fig:main_results}D, we show how the gap between classical and quantum runtime grows with $n$. For example, $n=\qubitshundredlogicalcyleone$ corresponds to $100\times$ speedup over classical with quantum runtime of \timehundredlogicalcyleone{} hours.

\begin{table*}
  \begin{tabular}{|>{\centering\arraybackslash}p{2.5cm}|c||>{\centering\arraybackslash}p{2.5cm}|>{\centering\arraybackslash}p{1.5cm}|>{\centering\arraybackslash}p{2cm}|>{\centering\arraybackslash}p{2cm}|>{\centering\arraybackslash}p{2.5cm}|}
    \cline{3-7}
    \multicolumn{1}{c}{}
    & \multicolumn{1}{c|}{}
    & \multirow{2}{=}{\centering \textbf{no improvements}}
    & \multirow{2}{=}{\centering smaller resource factories}
    & \multirow{2}{=}{\centering reduced logical cycle time}
    & \multirow{2}{=}{\centering $p_{\mathrm{phys}} = 10^{-4}$}
    & \multirow{2}{=}{\centering \textbf{combined improvements}}
    \\
    \multicolumn{1}{c}{} & \multicolumn{1}{c|}{} & & & & &
    \\ \cline{3-7} \noalign{\vskip\doublerulesep} \hline
    \multirow{3}{=}{\centering \textbf{realistic classical parallelization} \cite{arbelaez2013using}}
    & \textbf{problem size} $n$
    & \improvoptimisticcrossoverqubitsfactwospeedupfournumberthree{}
    & \improvoptimisticcrossoverqubitsfactwospeedupfournumbertwo{}
    & \improvoptimisticcrossoverqubitsfactwospeedupfournumberzero{}
    & \improvoptimisticcrossoverqubitsfactwospeedupfournumberone{}
    & \improvoptimisticcrossoverqubitsfactwospeedupfournumberfour{}
    \\ \cline{2-7}
    & \textbf{crossover time} $T_q$
    & \improvoptimistichoursfactwospeedupfournumberthree{} h
    & \improvoptimistichoursfactwospeedupfournumbertwo{} h
    & \improvoptimistichoursfactwospeedupfournumberzero{} h
    & \improvoptimistichoursfactwospeedupfournumberone{} h
    & \improvoptimistichoursfactwospeedupfournumberfour{} h
    \\ \cline{2-7}
    & \textbf{physical qubits ($\times 10^6$)}
    & \improvoptimisticTpqubitsfactwospeedupfournumberthree{}
    & \improvoptimisticTpqubitsfactwospeedupfournumbertwo{}
    & \improvoptimisticTpqubitsfactwospeedupfournumberzero{}
    & \improvoptimisticTpqubitsfactwospeedupfournumberone{}
    & \improvoptimisticTpqubitsfactwospeedupfournumberfour{}
    \\ \hline \hline
    \multirow{3}{=}{\centering \textbf{perfect classical parallelization}}
    & \textbf{problem size} $n$
    & \improvpessimisticcrossoverqubitsfactwospeedupfournumberthree{}
    & \improvpessimisticcrossoverqubitsfactwospeedupfournumbertwo{}
    & \improvpessimisticcrossoverqubitsfactwospeedupfournumberzero{}
    & \improvpessimisticcrossoverqubitsfactwospeedupfournumberone{}
    & \improvpessimisticcrossoverqubitsfactwospeedupfournumberfour{}
    \\ \cline{2-7}
    & \textbf{crossover time} $T_q$
    & \improvpessimisticdaysfaconespeedupfournumberthree{} d
    & \improvpessimisticdaysfaconespeedupfournumbertwo{} d
    & \improvpessimisticdaysfaconespeedupfournumberzero{} d
    & \improvpessimisticdaysfaconespeedupfournumberone{} d
    & \improvpessimistichoursfaconespeedupfournumberfour{} h
    \\ \cline{2-7}
    & \textbf{physical qubits ($\times 10^6$)}
    & \improvpessimisticTpqubitsfaconespeedupfournumberthree{}
    & \improvpessimisticTpqubitsfaconespeedupfournumbertwo{}
    & \improvpessimisticTpqubitsfaconespeedupfournumberzero{}
    & \improvpessimisticTpqubitsfaconespeedupfournumberone{}
    & \improvpessimisticTpqubitsfaconespeedupfournumberfour{}
    \\ \hline
  \end{tabular}
  \caption{\textbf{Impact of potential future advances on crossover against the MareNostrum 5 GPP supercomputer.}
    Here we consider the impact of the potential future improvements in quantum computing on estimated crossover points for 8-SAT with QAOA+AA with an asymptotically quartic speedup.
    Unlike our main results presented in Fig.~\ref{fig:main_results}, here we do not restrict the energy consumption of the classical solver and instead estimate the time-to-solution for the classical solver if it was run on all $725,760$ cores of MareNostrum 5 GPP (35th largest supercomputer in the world) \cite{top500}.
    For the classical algorithm, we consider both ``realistic'' parallelization of Sparrow based on the most parallelizable unstructured problem instance in Ref.~\cite{arbelaez2013using} (see \cref{sec:parallel_classical_solvers}), and perfect classical parallelization, which assumes that using $n_{\mathrm{cores}}$ CPUs speeds up the classical solver by a factor of $n_{\mathrm{cores}}$.
    For the quantum algorithm, we consider three potential improvements, namely: a fivefold reduction in the spacetime cost of a resource state, a fivefold reduction in the logical cycle time, and a tenfold reduction in physical error rates.
    As both sobering and aspirational points of reference, we also consider the possibility of achieving none or all these improvements.
    For each scenario, we report the number of variables ($n$), quantum runtime ($T_q$), and the required number of physical qubits (in millions) at the crossover point.
    We set the slowdown factor $\tau=2$ (discussed in Fig.~\ref{fig:main_results}) in the rows corresponding to realistic parallelization to reduce the memory footprint at the cost of increased quantum runtimes, and set $\tau=1$ for perfect classical parallelization.
  }
  \label{tab:improvements}
\end{table*}

\subsection{Opportunities for quantum resource reduction}

Our main results show that taking gate parallelization into account and carefully optimizing the quantum circuit for fault-tolerant execution can substantially reduce the crossover point for solving $8$-SAT using the QAOA+AA algorithm implemented on the surface code.
However, our analysis has so far relied only on existing techniques and conservative assumptions about hardware performance ($10^{-3}$ physical error rate and $1~\mu$s cycle time).
We now discuss the impact of potential future improvements to the surface code and hardware, focusing on three realistic improvements: reduced overheads for resource state factories via magic state cultivation \cite{gidney2024magic}, reduced logical cycle times through improved decoding and algorithmic fault-tolerance \cite{zhou2024algorithmic}, and reduced physical error rates.
As summarized in \cref{tab:improvements}, we estimate that these improvements enable a crossover time of only \improvoptimistichoursfactwospeedupfournumberfour{} hours against a classical solver parallelized over all $725,760$ CPU cores of MareNostrum 5 GPP, which is the 35th largest supercomputer in the world at the time of writing and one of the largest CPU-only systems, with total energy consumption of 5.75 MW~\cite{top500}.

\subsubsection{Better resource-state factories}

Recently, Gidney et al.~introduced magic state cultivation \cite{gidney2024magic} as an alternative to distillation.
This protocol leads to an order-of-magnitude reduction in the spacetime footprint required to produce a high-quality $T$ state.
At physical error rates of $10^{-3}$, however, none of the configurations considered in Ref.~\cite{gidney2024magic} achieve the requisite $T$ state fidelity ($\approx 10^{-12}$) for running QAOA+AA for 8-SAT.
Nonetheless, we note that the distillation protocols used for our main results are in fact two-stage distillation protocols, and one could imagine similarly performing magic-state cultivation as the first step in a multi-stage protocol to further improve the fidelity of a resource state.
Moreover, standalone cultivation is a new protocol with potential room for refinement, and it may be possible to achieve cultivation fidelities sufficiently high to obviate the need for secondary distillation.
With these considerations in mind, \cref{tab:improvements} shows the impact of reducing spacetime footprint of a resource state by a (modest) factor of 5 over the distillation-only approach considered in \cref{fig:main_results}.

\subsubsection{Reducing logical cycle times}

Our main results assumed a logical cycle time of $d\times 1~\mu$s, where $d$ is the code distance of the surface code.
This logical cycle time sets an effective clock speed for the fault-tolerant quantum computing architecture that we consider.
A clear way of speeding a quantum algorithm is to decrease this logical cycle time.
For example, one can imagine incorporating improved decoding schemes such as algorithmic fault tolerance \cite{zhou2024algorithmic} in order to reduce the requirement of performing $d$ syndrome measurement for every logical cycle.
The surface code error correction cycle time of $1~\mu$s may also be reduced due to algorithmic or hardware innovations, most notably reduced measurement times in superconducting platforms \cite{ye2024ultrafast}.
We consider the impact of reducing logical cycle times by a factor of 5 in \cref{tab:improvements}.

\subsubsection{Lower physical error rates}

Quantum hardware has matured remarkably over the last decade \cite{acharya2024quantum} and we expect this progress to continue. In particular, two-qubit error rates of $6\times 10^{-4}$ have been demonstrated in fluxonium superconducting qubits~\cite{2407.15783} and $3\times 10^{-4}$ in trapped ions~\cite{2407.07694}.
In \cref{tab:improvements}, we report how a reduction of physical error rates from $10^{-3}$ to $10^{-4}$ would impact the performance of QAOA+AA for 8-SAT.
We note that a reduction of this kind makes the cultivation protocol considered above sufficient in its current form to solve 8-SAT with QAOA+AA.
The results of \cite{gidney2024magic} indicate that physical error rates of $10^{-4}$ lead to $T$-states with fidelity of $4\times 10^{-11}$ and a 10-fold reduction in the spacetime cost of a $T$ state.
We consider the impact of reducing physical error rates to $10^{-4}$ (though still keeping the same distillation protocols that were used in our main results) in \cref{tab:improvements}.

\subsubsection{Additional directions}

In addition to the specific improvements discussed above, other developments and optimizations may yield additional quantum resource reductions.
We do not analyze these reductions in \cref{tab:improvements} as their impact is more difficult to quantify.
For example, further optimizing the circuit gadgets used in this work may lead to reduced time and ancilla-qubit overheads.
Aside from the production of resource states, most quantum processing in QAOA is devoted to implementing the single-qubit phase gate $P(\gamma)$.
Reducing the cost of implementing this gate in a surface code architecture---for example, by finding more gate-efficient or parallelizable implementations in a $\mathrm{Clifford}+T$ gate set \cite{kliuchnikov2024minimalentanglementinjectingdiagonal}, or by finding more efficient implementations in a different gate set \cite{kliuchnikov2023shorter, Gidney2019efficientmagicstate} that is compatible with the surface code architecture---may therefore yield substantial speedups for QAOA+AA.
Finally, there has been a great interest in recent years in the development of new quantum low-density-parity-checks codes \cite{Breuckmann_2021} and concatenated quantum codes \cite{yoshida2024, goto2024high}. Benefits of these codes include more favorable encoding rates and higher relative code distances which may reduce overheads for universal quantum computation \cite{Panteleev_2021, cowtan2025}.

\section{Discussion}
\label{sec:discussion}

Our result is the first complete and realistic resource estimate for a fault-tolerant quantum advantage in optimization. More generally, our findings indicate that low-degree polynomial speedups may be practical on realistic large-scale fault-tolerant quantum computers. This has outsized implications for domains like optimization, in which complexity-theoretic evidence suggests that broadly applicable exponential speedups are unlikely. However, we expect our findings to translate to other domains where low-degree polynomial speedups are available such as planted inference~\cite{2406.19378} and machine learning~\cite{Hastings2020}.

We focus our study on random satisfiability, which is well-studied and enables careful analysis. However, an important future direction is validating that the speedups we estimate also hold for hard industrially-relevant optimization problems, which are more challenging to characterize mathematically and study rigorously.

Our results rely on recent progress in surface code and decoding.
The parallelization of resource-state factories and careful scheduling of quantum subroutines, in particular, proved to be an effective strategy for minimizing computation time while avoiding large space overhead.
Continued advancements in quantum hardware and error correction can be expected to further reduce resource requirements for fault-tolerant quantum computation, and may thereby make practical the deployment of quantum algorithms for which a fault-tolerant quantum advantage was previously deemed implausible \cite{focus_beyond_quadratic}.

While the resource requirements we find are unlikely to be satisfied by near-term or early-fault-tolerant quantum processors, there are further optimizations one can consider.
In particular, we can simply run QAOA with very large $p$ such that $\Pr^{\rm success}_{\rm QAOA}$ is a large constant (e.g., $2/3$) and bypass amplitude amplification.
Ref.~\cite{2503.09563} shows that a simple parameter extrapolation procedure appears to lead to good QAOA parameters for large $p$, enabling execution at very large depth with no instance-specific parameter optimization.
Further investigation is need to ascertain the potential of this scheme and to understand its resource requirements.

There are two major limitations of our work.
First, we rely on the power-law decay of the asymptotic quantum scaling exponent with QAOA depth $p$, reported in Ref.~\cite{Sami_8_SAT_QAOA}.
The evidence in Ref.~\cite{Sami_8_SAT_QAOA} is for small QAOA depths $p$ and may not persist as $p$ grows.
Second, we do not consider the cost of routing
qubits.
The overhead of routing may become non-negligible given the large amount of parallelization in our quantum algorthm, though it may be overcome by careful optimization of the qubit layout and overall fault-tolerant architecture.

\section{Methods}
\label{sec:methods}

\subsection{Amplitude amplification with unknown initial success probability}

Here we present a procedure inspired by Ref.~\cite{brassard2000quantum} for ensuring that we do not overshoot in the amplitude amplification process. The overhead comes from needing to repeat the amplification process multiple times to account for uncertainty in the success probability of QAOA.

\begin{theorem}
  \label{th:quantum_overshooting}
  Let $P_{\text{SAT}}$ be the projector on the space of satisfying assignments and $U_{\text{SAT}}$ an oracle for flagging satisfying assignments. Suppose $|\psi\rangle = U_{\psi}|\mathbf{0}\rangle$ is a quantum state satisfying $p = \lVert P_{\text{SAT}}|\psi\rangle\rVert_2^2$. Then there is a quantum algorithm that with probability at least $1- \delta$, makes at most
  \begin{align*}
    \lceil \frac{\pi}{4\sqrt{p}}\log_2(1/\delta)\rceil
  \end{align*} queries to $U_{\psi}$, $U_{\psi}^{\dagger}$, $U_{\text{SAT}}$, and $U_{\text{SAT}}^{\dagger}$ to output a satisfying assignment.
\end{theorem}
\begin{proof}
  Recall that if $\theta_{a} = \sin^{-1}(\lVert P_{\text{SAT}}|\psi\rangle\rVert_2)$ for a state $|\psi\rangle$, then $m$ rounds of amplitude amplification with $U_{\psi}$ and $U_{\text{SAT}}$ boosts the amplitude to $\sin((2m+1)\theta_a)$ \cite{brassard2000quantum}.

  Begin by measuring $|\psi\rangle$ in the computational basis $\lceil \log(1/\delta)\rceil$ times. If $ p \in [\frac{1}{2}, 1]$, then we will find a satisfying assignment with probability at least $1- \delta$.  If no solution is found, then take $\theta = \sin^{-1}(\frac{1}{\sqrt{2}})$ and perform
  \begin{align}
    \lceil \frac{\pi}{8\theta} - \frac{1}{2}\rceil
  \end{align}
  round of amplitude amplification and then measure the result. This will ensure that if $\theta_a \in [\frac{\theta}{2}, \theta]$, then the probability of success will be boosted to at least $\frac{1}{2}$. We then repeat this $\log_2(1/\delta)$ times, which ensures that if $\theta_{a} \in [\frac{\theta}{2}, \theta]$, we will see a satisfying bitstring with probability at least $1-\delta$.
  If we get a satisfying assignment, then we stop. Otherwise, we continue to halve $\theta$ and repeat the above until a satisfying bit string is found.

  Let $k = \lfloor \log_2(\theta/\theta_a) \rfloor$.
  With probability at least $1-\delta$, the number of calls to the unitaries specified in theorem will be at most

  \begin{align*}
    \lceil\log_2(1/\delta)\rceil & \cdot \sum_{i=0}^{k} \lceil \frac{2^i\pi}{8\theta} - \frac{1}{2}\rceil \leq  \\
    & \leq  \lceil\log_2(1/\delta)\rceil\sum_{i=0}^{k} \lceil \frac{\pi}{8\sqrt{p}\cdot 2^{k-i}} - \frac{1}{2}\rceil \\
    &\leq  \lceil\log_2(1/\delta)\rceil\lceil \frac{\pi}{4\sqrt{p}} \rceil.
  \end{align*}
  Note only the interval containing $\theta_a$ needs to succeed for us to terminate successfully in the above specified time.
\end{proof}
We set $\delta = 1/2^4=0.0625$ for the analysis in the main text, which corresponds to an overhead of $\lceil\log_2(1/\delta)\rceil = 4$ over ordinary Grover.

\subsection{Optimal decomposition accuracy for phase gates}
\label{sec:optimal_decomp_arb_rotation}

In the context fault-tolerant quantum computing, implementing arbitrary rotations is challenging due to the constraints of the fault-tolerant universal gate set for a given architecture.
Typically, these gate sets include the Clifford group gates (such as the Hadamard, Phase, and CNOT gates) and a non-Clifford gate such as the $T$-gate.
While the Clifford gates are relatively easy to implement fault-tolerantly, they are not sufficient for universal quantum computation.
The inclusion of the $T$-gate allows for universality, enabling the approximation of any quantum operation to arbitrary precision at the cost of additional overhead such as magic state distillation.

The Solovay-Kitaev theorem provides a general method for approximating any single-qubit unitary operation using a finite gate set \cite{dawson2005solovay}.
While this method is general, it is not always the most efficient in practice
\cite{ross2016optimal, kliuchnikov2023shorter, Bocharov_efficient_synthesis_2015}.
For a fixed decomposition accuracy $\delta$, the required number of $T$ gates, $\mathcal{N}_T$, typically scales as
\begin{align}
  \mathcal{N}_T = b \log_2\delta^{-1} + c,
  \label{eq:T_gate_count}
\end{align}
where $b,c$ are positive real numbers whose precise value depends on the decomposition scheme \cite{ross2016optimal, Bocharov_efficient_synthesis_2015}.
In this section, we consider two decompositions schemes with $(b,c)=(3,9.19)$ and $(0.57,8.83)$, which are studied in detail in Refs.~\cite{ross2016optimal} and \cite{kliuchnikov2023shorter}, respectively.

We can decompose a phase gate into a sequence of $X$ and $Z$ rotations, where each rotation consumes $T$ states and addresses the rotated qubit for one logical cycle \cite{litinski2022active}.
Treating the infidelity $\epsilon_T$ of a faulty $T$ state as a probability of a depolarization error after consuming this $T$ state, we combine the infidelity after $\mathcal{N}_T$ depolarization channels (see \cref{sec:entanglement_fidelity}) with the decomposition error $\delta$ of a phase gate to obtain an overall entanglement infidelity of decomposition,
\begin{align}
  \mathcal{I}_{\mathrm{ent}}(\epsilon_T, \delta)
  =
  1 - \frac{1}{4}(1+3(1-\epsilon_T)
  ^{\mathcal{N}_T})
  \left (1-\delta^2\right).
  \label{eq:infidelity_cal}
\end{align}
For a fixed $T$ gate infidelity $\epsilon_T$, the optimal decomposition accuracy is
$\delta_{\mathrm{opt}}=\argmin_\delta \mathcal{I}_{\mathrm{ent}}(\epsilon_T,\delta)$.
Substituting Eq.~\eqref{eq:T_gate_count} into Eq.~\eqref{eq:infidelity_cal} and setting $\partial \mathcal{I}_{\mathrm{ent}}(\epsilon_T, \delta)/\partial\delta=0$, we find that the optimal decomposition accuracy $\delta_{\mathrm{opt}}$ satisfies
\begin{multline}
  2\delta_{\mathrm{opt}}^2 \left(1+3(1-\epsilon_T)^{b\log_2\delta_{\mathrm{opt}}^{-1} + c}\right) \\
  =
  -3 b (1-\epsilon_T)^{b\log_2\delta_{\mathrm{opt}}^{-1}+c}
  \log_2(1-\epsilon_T)\left(1-{\delta_{\mathrm{opt}}^2}\right).
  \label{eq:exact_decomp_fidelity}
\end{multline}
This condition is a transcendental equation that can be solved numerically to compute $\delta_{\mathrm{opt}}$ for any given value of $\epsilon_T$.
As an estimate, however, in the regime $\epsilon_T,\delta_{\mathrm{opt}}\ll 1$ we can simplify
\begin{align}
  \delta_{\mathrm{opt}}
  \approx \sqrt{\left(\frac{3b\epsilon_T}{8\ln 2}\right)}.
  \label{eq:approx_decomp_fidelity}
\end{align}

To better understand the impact of decomposing arbitrary rotations on the fidelity of a quantum circuit, we consider the following.
We have a circuit with  a total of  $G$  rotation gates, and we aim to achieve a final infidelity of the circuit $\mathcal{I}$.
Thus one can write,
\begin{align}
  \mathcal{I}\geq 1-\left(1-\mathcal{I}(\epsilon_T,\delta_{\mathrm{opt}})\right)^{G}.
  \label{eq:get_target_t_infidelity_depth}
\end{align}
From this one can obtain a maximum value of $\epsilon_T$ that achieves a target infidelity for the circuit.
\cref{fig:T_fidelity_depth} shows the $T$ state fidelity required to achieve a target QAOA+AA circuit infidelities $\mathcal{I}$ for different numbers of phase gates $G$.
We consider two different decomposition schemes and two different target entanglement fidelities for the calculations. Fig.~\ref{fig:T_fidelity_depth} suggests a significant dependence on the number of gates $G$ and target accuracy $1-\mathcal{I}$ in determining the needed $T$ state fidelity $1-\epsilon_T$.

\begin{figure}
  \centering
  \includegraphics{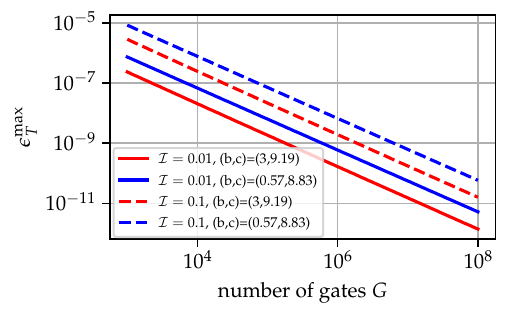}
  \caption{
    The desired infidelity of the $T$ state to achieve a target QAOA+AA circuit infidelity as a function of the total number of phase gates ($G$) in \cref{eq:get_target_t_infidelity_depth}.
    We consider two different decomposition schemes and two different target infidelities $\mathcal{I}$ for the simulations.
    The plot suggests a notable dependence on the number of gates and target accuracy (entanglement fidelity) in determining the needed $T$ state fidelity.
  }
  \label{fig:T_fidelity_depth}
\end{figure}

\subsection{Resource requirements for the $k$-SAT QAOA phaser}
\label{sec:phase_op_resources}

Here we consider the spacetime cost of applying the QAOA phaser for 8-SAT, which is the dominant cost of the $k$-SAT QAOA+AA circuit as sketched out in Fig.~\ref{fig:algorithms} of the main text.
The $k$-SAT phaser is
\begin{align}
  U_C(\gamma) &= e^{-i\gamma H_C} = \prod_{j=1}^m U_j(\gamma),
\end{align}
where
\begin{align}
  U_j(\gamma) &= \sum_{x\in\set{0,1}^n} e^{-i\gamma C_j(x)} \op{x},
  \label{eq:phase-op-cost}
\end{align}
is an exponentiated clause that can be converted into a $k$-qubit phase gate $P_k(\gamma) = e^{i\gamma\op{1}^{\otimes k}}$ by appropriately conjugating qubits that are negated in $C_j$ by Pauli-$X$ gates.
For brevity, we henceforth keep the value of $\gamma$ arbitrary but fixed, and suppress the explicit dependence of $U_C(\gamma)$, $U_j(\gamma)$, and derived objects on $\gamma$.

We say that a circuit $V$ is $Z$-equivalent to the unitary $U$ if the action of $U$ is equal to the action of $V$ followed by single-qubit Pauli-$Z$ corrections that can be efficiently computed from the outcomes of measurements performed in $V$.
For the purposes of the statements below, we define a \textit{logical cycle} to be the time required to perform one logical two-qubit Pauli operator measurement.
In practice, one logical cycle in a distance-$d$ surface code on a square lattice with nearest-neighbor interactions is the time required to perform $d$ rounds of syndrome measurement \cite{litinski2019game}.
The key technical result that we leverage to bound the resource requirements of the $k$-SAT phaser $U_C$ is the following:
\begin{theorem}
  The $K$-qubit phase gate $P_K(\gamma) = e^{i\gamma\op{1}^{\otimes K}}$ is $Z$-equivalent to a circuit that addresses the $K$ qubits of $P_K(\gamma)$ for one logical cycle, introduces $\lfloor\frac{13}{2} K\rfloor$ ancilla qubits, consumes $\sum_{\ell=1}^{\lceil\log_2 K\rceil} \lfloor K/2^\ell\rfloor \le K-1$ CCZ states, and can be implemented in a total of $n_P + 4\lceil\log_2 K\rceil$ logical cycles, where $n_P$ is the number of logical cycles required to implement the single-qubit phase gate $P_1(\gamma)$.
  \label{thm:phase-gate}
\end{theorem}
We prove this theorem by construction in \cref{sec:TACU}.

Let $Q_j$ be the set of qubits that correspond to the variables addressed in the clause $C_j$.
Since all exponentiated clauses $U_j$ are diagonal in the computational basis, in order to implement the phaser $U_C$ it is sufficient to implement each exponentiated clause $U_j$ by a $Z$-equivalent circuit $V_j$, and defer all Pauli-$Z$ corrections to the end of $V_C = \prod_j V_j$.
We refer to the circuit $V_j$ as a \textit{task}, and we say that task $V_j$ has been \textit{dispatched} once all of the the gates in $V_j$ that address qubits in $Q_j$ have been applied.

Let $\mathop{\mathrm{LogicalDepth}}(W)$ denote the total number of logical cycles required to implement the circuit $W$.
We say that a partition $\mathcal{C} = \set{\mathcal{C}_1,\mathcal{C}_2,\cdots,\mathcal{C}_c}$ of the $k$-SAT clauses $C = \set{C_1,\cdots,C_m}$ is \textit{disjoint} if each part $\mathcal{C}_\ell\subset C$ consists of clauses that address mutually disjoint sets of qubits, which is to say that if $C_i,C_j\in\mathcal{C}_\ell$, then $\abs{Q_i\cap Q_j} = 0$.
Theorem \ref{thm:phase-gate} implies that
\begin{corollary}
  A circuit $V_C = \prod_j V_j$ in which $V_j$ is $Z$-equivalent to the exponentiated $k-\mathrm{SAT}$ clause $U_j$ can be implemented in $\mathop{\mathrm{LogicalDepth}}(V_C)\le c + n_P + 4\lceil\log_2 k\rceil$ logical cycles, where $c$ is the size of a disjoint partition partition of $C$.
  \label{thm:phaser-depth}
\end{corollary}
The basic idea here is that a disjoint partition $\mathcal{C}$ of $C$ induces a \textit{schedule} $\mathcal{V}_{\mathcal{C}} = (\mathcal{V}_1, \mathcal{V}_2, \cdots, \mathcal{V}_c)$, where $\mathcal{V}_\ell = \set{V_i:C_i\in\mathcal{C}_\ell}$ is a collection of tasks that can be executed in parallel in one logical cycle.
It therefore takes $c$ logical cycles to sequentially dispatch all tasks in $V_C$ according to the schedule $\mathcal{V}_{\mathcal{C}}$.
The additional contribution of $n_P + 4\lceil\log_2 k\rceil$ in Corollary \ref{thm:phaser-depth} accounts for the time required for the first task prior to dispatching, and the last task after dispatching.
For the 8-SAT instances considered in the main text, $c \gg n_P + 4\lceil\log_2 k\rceil$, so $\mathop{\mathrm{LogicalDepth}}(V_C)\approx c\approx 12r = 2112$.

\subsection{Target distance of the surface code}
\label{sec:distance_surface_code}

To determine the target distance of the surface code for the QAOA+AA algorithm applied to 8-SAT, we consider the physical error rate $p_{\mathrm{ph}}$ and the total number of non-Clifford gates $G$ in the circuit.
The final infidelity of $\mathcal{I}$ the circuit is conservatively bounded by:
\begin{align}
  \mathcal{I} \leq 1 - \mathcal{F}(p_{\mathrm{ph}})^G
\end{align}
where the fidelity function $ \mathcal{F}(p_{\mathrm{ph}}) $ is approximately given by:
\begin{align}
  \mathcal{F}(p_{\mathrm{ph}})
  \approx 1 - \left(\frac{p_{\mathrm{ph}}}{p_{\mathrm{th}}}\right)^{d/2}.
\end{align}
Here $ p_{\mathrm{th}} \approx 10^{-2}$  is the threshold error rate for the surface code.
For a target infidelity $\mathcal{I}_{\mathrm{tar}}$, we therefore solve these expressions for $d$ in the limits that $\mathcal{I}_{\mathrm{tar}}, 1 - \mathcal{F}(p_{\mathrm{ph}})\ll1$, and set the code distance to
\begin{align}
  d = \left\lceil\frac{2\log(\mathcal{I}_{\mathrm{tar}}/G)}{\log(p_{\mathrm{ph}}/p_{\mathrm{th}})}\right\rceil.
\end{align}
We set $\mathcal{I}_{\mathrm{tar}}=0.99$ and $p_{\mathrm{ph}}/p_{\mathrm{th}} = 0.1$ in the main text.

\subsection{Number of surface code decoders}
\label{sec:number_of_decoders}

To determine the number of surface-code decoders required by the QAOA+AA algorithm, we first assign one decoder to each logical data and ancilla qubit.
Recognizing that resource-state factories may require more complex processing, as a heuristic we assign 10 decoders to each resource-state factory, in total setting
\begin{align}
  n_{\mathrm{decoders}}
  = n + n_{\mathrm{ancilla}} + 10\times n_{\mathrm{fac}}.
\end{align}

\subsection{Classical solvers for random 8-SAT}
\label{sec:classical_solvers}

\paragraph{SAT and random k-SAT}
As the first problem known to be NP-complete \cite{karp2009reducibility}, the Boolean satisfiability problem (SAT) is a fundamental challenge in mathematics and computer science, with wide-ranging applications in hardware
verification \cite{gupta2006sat}, motion planning \cite{imeson2019motionplanning}, discrete optimization \cite{heule2016solving}, machine learning \cite{baluta2019quantitativesatml} and quantum computing \cite{vardi2023solving}.

A SAT formula is conventionally represented in the conjunctive normal form (CNF). A (CNF) k-SAT problem instance on $n$ Boolean variables $x_1,\cdots, x_n$ is a conjunction (\texttt{AND}, $\wedge$) of $m$ clauses, where each clause is a disjunction (\texttt{OR}, $\vee$) of $k$ literals and a literal is a Boolean variable or its negation. An example for a 3-SAT formula with $n=4$, $m=2$ is $(x_1\vee x_2\vee x_4)\wedge (\neg x_2\vee x_3\vee\neg x_4)$. A clause is satisfied if at least one of its literals is. A formula is satisfied by an assignment of Boolean variables if all of its clauses are satisfied. The formula above is satisfied by $x_1=x_2=x_3=x_4=1$.

Random (CNF) $k$-SAT is an important and well-studied family of SAT benchmarks, where each clause consists of $k$ random literals.
The clause-variable ratio, denoted by $\alpha$, plays a critical role in random $k$-SAT problems.
It is observed experimentally that the random $k$-SAT problem goes through a satisfiability phase-transition as $\alpha$ increases \cite{crawford1996experimental}.
It is conjectured that the phase-transition is sharp, which is to say that if $f_{n,k,\alpha}$ is a random $k$-SAT formula with $n$ variables and $\alpha n$ clauses, then there exists a constant $\alpha_k$ for which the probability $\mathbb{P}[\operatorname{SAT}(f_{n,k,\alpha})]$ that $f_{n,k,\alpha}$ is satisfiable behaves as
\begin{align*}
  \lim\limits_{n\to\infty}\mathbb{P}[\operatorname{SAT}(f_{n,k,\alpha})]
  =
  \begin{cases}
    1 & \alpha < \alpha_k \\
    0 & \alpha > \alpha_k
  \end{cases}.
\end{align*}
The conjuncture has been proven for $k=2$ \cite{goerdt1992threshold,reedmick} and large $k$ by the landmark work in Ref.~\cite{ding2015proof}.
Experiments indicate that formulas close to the threshold $\alpha_k$ are especially challenging for practical SAT solvers \cite{bresler2022algorithmic}.

\paragraph{Benchmark generation} This work focuses on random $8$-CNF benchmarks near the satisfiability threshold $\alpha_8$, which is experimentally observed to be approximately $176$ \cite{Sami_8_SAT_QAOA}.   Each problem is formed by generating $176n$ random $OR$-clauses where each clause consists of $8$ literals chosen uniformly (with replacement) from $\set{x_1,\cdots,x_n}$ and negated with probability $1/2$. For each instance generated, we use the complete solver Kissat \cite{heule2024proceedings} to prove that it is satisfiable; if the instance is not satisfiable, it is discarded. For each value of $n$ ranging from 30, 31, $\ldots$, 70, at least 70 satisfiable instances were generated.

\paragraph{Experiment configurations} The time limit is set to 1000 seconds for all problem instances. The experiment is conducted on an AMD EPYC 7R32 CPU with 48 cores 
and a TDP of 280W.

\paragraph{Solvers Selection} SAT solvers are typically classified into complete and incomplete solvers. Complete solvers are able to determine the satisfiability of such formulas by either giving a satisfying assignment or proving unsatisfiability. State-of-the-art complete SAT solvers based on Conflict-Driven Clause Learning (CDCL) \cite{biere2009conflict}, e.g., Kissat \cite{heule2024proceedings} and MapleSAT \cite{biere2017cadical}, are dominant in solving industrial SAT instances.  Incomplete  solvers based on (stochastic) local search (SLS)
\cite{hoos1999stochastic,kyrillidis2021continuous} and message passing (MP) \cite{kroc2009message} offer a practical alternative by conducting unstructed search, moving quickly in the Boolean cube to decrease the number of unsatisfied clauses.  They haven shown promise in promptly solving satisfiable large random instances \cite{lorenz2020effect}.

We included five state-of-the-art incomplete SAT solvers 
in the evaluation of classic SAT solvers:

\begin{itemize}
  \item probSAT~\cite{balint2012choosing}, a local search based SAT solver built on WalkSAT with probabilistic variable selection heuristics. %
  \item yalSAT~\cite{biere2016splatz}, the winner of the random track of SAT Competition (SC) 2017~\cite{balyosat} built based on probSAT.
  \item Dimetheus~\cite{gableske2013performance}, a solver combining SLS and MP techniques and the winner of the random track in SC 2016~\cite{balyo2017sat}.
  \item Sparrow~\cite{balint2014engineering}, the winner of the random track of SC 2018~\cite{heule2019sat}.
  \item WalkSATlm~\cite{cai2013improving}, a solver combines the classic WalkSAT solver with tie-breaking heuristics.
\end{itemize}

We also considered the following SAT solvers but did not include them in the evaluation, including SOTA complete solvers MapleSAT and Kissat, incomplete solvers WalkSAT \cite{selman1993local} and Survey Propagation \cite{braunstein2005survey}.
They are, however, outperformed by the solvers in the evaluation by a large margin.

\subsection{Parallelizing classical solvers for 8-SAT}
\label{sec:parallel_classical_solvers}

The local search solvers that are most performant for random SAT can be straightforwardly parallelized by running many independent local searches in parallel.
Since no communication is required between independent local searches, the runtime for a parallelized search can be predicted by analyzing the distribution of runtimes for one local search.
Specifically, if $n_{\rm cores}$ searches are executed in parallel, they can all be terminated as soon as one of them finishes, so the expected runtime for $n_{\rm cores}$ searches is the minimum over $n_{\rm cores}$ samples from the distribution of runtimes for a single search.

Ref.~\cite{arbelaez2013using} uses this insight to analyze the performance of Sparrow and predict how it scales to large supercomputers.
Specifically, they benchmark Sparrow on random instances from SAT competitions and find that the running times follow shifted exponential distribution for easier random instances and lognormal for harder ones. 
To make a conservative assumption for our crossover point, we choose the random instance with the most favorable classical scaling from Ref.~\cite{arbelaez2013using}, namely Rand-9.
The distributions of runtimes for this instances is the shifted exponential distribution $f_Y(t)=\lambda e^{-\lambda(t-x_0)}$ with parameters $\lambda=9.8\times 10^{-4}$ and $x_0=6.8$.
We then use the model from Ref.~\cite{arbelaez2013using} to obtain a speedup from parallelization with $n_{\rm cores}$.
Finally, we divide the measured serial runtime of Sparrow by this speedup factor, obtaining the estimated parallel runtime.

\section*{Acknowledgements}

The authors thank Pradeep Niroula for helpful discussions about the overhead of fixed-point amplitude amplification. The authors thank their colleagues at the Global Technology Applied Research center of JPMorganChase for support.

\section*{Data availability}
The data for this work is available at \href{https://doi.org/10.5281/zenodo.15122122}{10.5281/zenodo.15122122}.

\bibliography{main}

\section*{Disclaimer}
This paper was prepared for informational purposes by the Global Technology Applied Research center of JPMorganChase. This paper is not a product of the Research Department of JPMorganChase or its affiliates. Neither JPMorganChase nor any of its affiliates makes any explicit or implied representation or warranty and none of them accept any liability in connection with this position paper, including, without limitation, with respect to the completeness, accuracy, or reliability of the information contained herein and the potential legal, compliance, tax, or accounting effects thereof. This document is not intended as investment research or investment advice, or as a recommendation, offer, or solicitation for the purchase or sale of any security, financial instrument, financial product or service, or to be used in any way for evaluating the merits of participating in any transaction.

\appendix
\onecolumngrid

\section{Entanglement Fidelity}
\label{sec:entanglement_fidelity}

Entanglement fidelity measures how well the state of a system and its entanglement with other systems is preserved under the action of a quantum channel.
The sensitivity of entanglement fidelity to the \textit{entanglement} between systems makes it a better measure of how well a state is preserved under the action of a channel than other measures, such as gate fidelity, in the context of error-corrected quantum computation \cite{nielsen1996entanglement}.
Mathematically, the entanglement fidelity $F_{\mathrm{ent}}(\mathcal{E})$ of a quantum channel $\mathcal{E}$ that acts on density operators of a $d$-dimensional Hilbert space is \cite{schumacher1996entanglement}
\begin{align}
  F_{\mathrm{ent}}(\mathcal{E})
  =  \braket{\phi | (\mathcal{E} \otimes \mathcal{I})(\op\phi) | \phi}
  = \Tr\left[\Pi_\phi (\mathcal{E} \otimes \mathcal{I})(\Pi_\phi)\right],
  \label{eq:ent_fidelity}
\end{align}
where $\ket\phi = \frac{1}{\sqrt{d}} \sum_{a=0}^{d-1} \ket{a}\otimes\ket{a}$ is a maximally entangled state of two $d$-dimensional systems, $\Pi_\phi = \op{\phi}$ is a projector onto $\ket\phi$, and $\mathcal{I}$ is the identity channel.
When a fixed channel acts $\mathcal{E}$ on a particular state $\rho$, we may refer to $F_{\mathrm{ent}}(\mathcal{E})$ as the fidelity of $\rho$ after the application of $\mathcal{E}$.

Consider the depolarizing channel $\mathcal{D}_p$ that replaces its input by the maximally mixed state $\sigma = \frac1d \sum_{a=0}^{d-1} \op{a}$ with probability $p$,
\begin{align}
  \mathcal{D}_p(\rho) = (1-p) \rho + p \sigma.
\end{align}
\begin{theorem}
  The entanglement fidelity of a quantum state after $n$ applications of the depolarizing channel $\mathcal{D}_p$ is
  \begin{align}
    F_{\mathrm{ent}}(\mathcal{D}_p^n) = \left(1 - \frac{1}{d^2}\right) (1-p)^n + \frac1{d^2}.
  \end{align}
\end{theorem}
\textit{Proof}.
Expanding $\Pi_\phi = \frac1d \sum_{a,b=0}^{d-1} \op{a}{b}\otimes\op{a}{b}$, the action of $\mathcal{D}_{p,0} = \mathcal{D}_p\otimes\mathcal{I}$ on $\Pi_\phi$ is
\begin{align}
  \mathcal{D}_{p,0}(\Pi_\phi)
  = (1-p) \Pi_\phi + p \times \frac1d \sum_{a,b} \sigma \otimes \op{a}{b}
  = (1-p) \Pi_\phi + p \sigma \otimes \omega
\end{align}
where $\omega = \op{+}$ is the density operator of the uniform superposition $\ket{+} = \frac1{\sqrt{d}} \sum_{a=0}^{d-1}\ket{a}$.
The depolarizing channel acts trivially on an already depolarized input, which is to say that
\begin{align}
  \mathcal{D}_{p,0}(\sigma \otimes \omega) = \sigma \otimes \omega.
\end{align}
Applying the depolarizing channel $n$ times to the maximally mixed state $\Pi_\phi$ thus yields the state
\begin{align}
  \mathcal{D}_{p,0}^n(\Pi_\phi) = (1-p)^n \Pi_\phi + q_n(p) \sigma \otimes \omega,
\end{align}
where the coefficient $q_n(p)$ can be found be enforcing that the $n$-fold depolarized state has trace 1,
\begin{align}
  1 = \Tr[\mathcal{D}_{p,0}^n(\Pi_\phi)] = (1-p)^n + q_n(p),
\end{align}
which implies that $q_n(p) = 1 - (1-p)^n$.
Altogether, the entanglement fidelity of the $n$-fold depolarization channel is
\begin{align}
  F_{\mathrm{ent}}(\mathcal{D}_p^n)
  = \Tr\left[\Pi_\phi \mathcal{D}_{p,0}(\Pi_\phi)\right]
  = (1-p)^n + q_n(p) \Tr\left[\Pi_\phi (\sigma \otimes \omega)\right],
\end{align}
where
\begin{align}
  \Tr\left[\Pi_\phi (\sigma \otimes \omega)\right]
  &= \frac1{d^3} \sum_{a,b,i,j,k} \Tr[(\op{a}{b} \otimes\op{a}{b})(\op{i}\otimes\op{j}{k})] \\
  &= \frac1{d^3} \sum_{a,b,i,j,k} \Tr[\op{a}{b} \op{i}] \times \Tr[\op{a}{b}\op{j}{k}] \\
  &= \frac1{d^2},
\end{align}
so
\begin{align}
  F_{\mathrm{ent}}(\mathcal{D}_p^n) = (1-p)^n + (1 - (1-p)^n) \frac1{d^2}
  = \frac{d^2-1}{d^2} (1-p)^n + \frac1{d^2}.
\end{align}

\subsection{Entanglement Fidelity and Operator Norm}

It is common in the (deterministic) unitary gate synthesis literature to consider the operator-norm distance $D(U,V) = \norm{U - V}$ as a measure of similarity between two unitaries $U$ and $V$, where $\norm{W}$ is the largest singular value of $W$ \cite{selinger2012efficient, ross2016optimal, hastings2016turning, wiebe2014quantum}.
When synthesizing a gate up to global phase, it is convenient to instead consider the phase-agnostic distance
\begin{align}
  \tilde D(U, V)
  = \min_\alpha D(U, e^{i\alpha} V)
  = \min_\alpha \norm{U - e^{i\alpha} V}.
\end{align}
Here we establish the relationship between the phase-agnostic operator-norm distance $\tilde D(U, V)$ and the entanglement fidelity
\begin{align}
  \mathcal{F}_{\mathrm{ent}}(U^\dag V) = \frac1{d^2} \abs{\Tr(U^\dag V)}^2.
\end{align}
We say that $U$ and $V$ are $\delta$-close if $\tilde D(U, V)=O(\delta)$ and $\delta<1$.
\begin{theorem}
  The entanglement fidelity $\mathcal{F}_{\mathrm{ent}}(U^\dag V)$ and the phase-agnostic distance $\tilde D(U, V)$ between two $\delta$-close single-qubit unitaries $U,V$ are related by
  \begin{align}
    \mathcal{F}_{\mathrm{ent}}(U^\dag V) = 1 - \tilde D(U,V)^2 + O(\delta^4).
  \end{align}
  \label{thm:fidelity-distance}
\end{theorem}
\textit{Proof}.
Without loss of generality we can expand, for some angles $\phi,\theta\in[-\pi,\pi]$,
\begin{align}
  U^\dag V = \exp(i\phi + i\frac{\theta}{2}\vec v\cdot\vec X),
\end{align}
where $\vec v = (v_x, v_y, v_z)$ is a unit vector with $v_x^2 + v_y^2 + v_z^2 = 1$ and $\vec X=(X,Y,Z)$ is a vector of Pauli operators.
This expansion allows us to directly compute
\begin{align}
  \mathcal{F}_{\mathrm{ent}}(U^\dag V)
  = \frac14 \abs{\Tr(e^{i\phi + i\frac{\theta}{2}\vec v\cdot\vec X})}^2
  = \frac14 \abs{\Tr(e^{i\frac{\theta}{2}\vec v\cdot\vec X})}^2
  = \cos(\frac{\theta}{2})^2.
\end{align}
We then observe that
\begin{align}
  \tilde D(U,V)
  = \min_\alpha \norm{\mathds{1} - e^{i\alpha} U^\dag V}
  = \min_\alpha \norm{\lambda_\alpha^+ \Pi_{+v} + \lambda_\alpha^- \Pi_{-v}},
\end{align}
where $\lambda_\alpha^\pm = 1 - e^{i\alpha + i\phi \pm i\theta/2}$ and $\Pi_{\pm v} = \op{\pm v}$ are projectors onto the orthonormal eigenvectors $\ket{\pm v}$ of $\vec v\cdot\vec X$.
It follows that
\begin{align}
  \tilde D(U,V)
  &= 2 \min_\alpha \max_\pm \abs{\lambda_\alpha^\pm} \\
  &= 2 \min_\alpha \max_\pm \abs{\sin(\frac{\alpha + \phi \pm \theta/2}{2})} \\
  &= 2 \min_\beta \max_\pm \abs{\sin(\beta \pm \frac{\theta}{4})} \\
  &= 2 \min_\beta \max_\pm
  \left(
    \abs{\sin\beta}\, \abs{\cos(\frac{\theta}{4})}
    \pm \abs{\cos\beta}\, \abs{\sin(\frac{\theta}{4})}
  \right) \\
  &= 2 \min_\beta
  \left(
    \abs{\sin\beta}\, \abs{\cos(\frac{\theta}{4})} + \abs{\cos\beta}\, \abs{\sin(\frac{\theta}{4})}
  \right).
  \label{eq:DUV_beta}
\end{align}
Without loss of generality, we can restrict $\beta\in[0,\pi/2]$, and consider the quantity $f(\beta) = \sin\beta\abs{\cos(\frac{\theta}{4})} + \cos\beta \abs{\sin(\frac{\theta}{4})}$ that is minimized over $\beta$ in Eq.~\eqref{eq:DUV_beta}.
The derivative $\partial_\beta f(\beta) = \cos\beta \abs{\cos(\frac{\theta}{4})} - \sin\beta \abs{\sin(\frac{\theta}{4})}$ is positive at $\beta=0$, negative at $\beta=\pi/2$, and zero once in between.
It follows that $f(\beta)$ achieves a single maximum at some $\beta\in(0,\pi/2)$, and is otherwise minimal at one of its endpoints, $f(0)=\abs{\sin(\theta/4)}$ or $f(\pi/2)=\abs{\cos(\theta/4)}$.
The fact that $\abs{\theta}\le\pi$ then implies that $\abs{\sin(\theta/4)} \le \abs{\cos(\theta/4)}$, so
\begin{align}
  \tilde D(U,V) = 2\abs{\sin(\frac{\theta}{4})}.
\end{align}
Altogether, for positive $\theta=\delta<1$ we can expand
\begin{align}
  \tilde D(U,V)
  = 2 \sin(\frac{\delta}{4})
  = \frac{\delta}{2} + O(\delta^3),
\end{align}
and
\begin{align}
  \mathcal{F}_{\mathrm{ent}}(U^\dag V)
  = 1 - \frac{\delta^2}{4} + O(\delta^4)
  = 1 - \tilde D(U,V)^2 + O(\delta^4),
\end{align}
thereby arriving at Theorem \ref{thm:fidelity-distance}.

\section{The TACU gadget}
\label{sec:TACU}

An essential circuit primitive for both the Hamiltonian simulation and oracles of the QAOA+AA algorithm in this work is the $K$-qubit temporary-AND-compute-and-uncompute (TACU) gadget \cite{litinski2022active}, which we use to apply $K$-qubit phase gates of the form $P_K(\gamma) = e^{-i\gamma\op{1}^{\otimes K}}$.
Specifically, the $K$-qubit TACU gadget is defined by
\begin{align}
  \vcenter{\hbox{\includegraphics{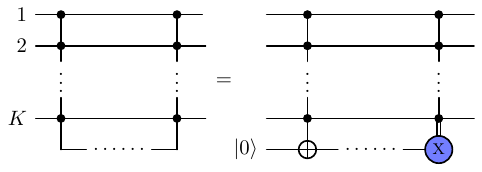}}},
\end{align}
or, for shorthand,
\begin{align}
  \vcenter{\hbox{\includegraphics{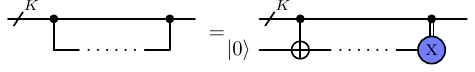}}},
  \label{eq:tacu-k-def}
\end{align}
where the horizontal dots ($\ldots\ldots$) are a placeholder for ancilla-qubit operations, the first gate on the right-hand side is a $(K+1)$-qubit Toffoli gate that applies at Pauli-$X$ to the ancilla qubit if all controls are in $\ket{1}$, and the last gate denotes an $X$-basis measurement of the ancilla qubit, whose measurement outcome determines whether to apply a $K$-qubit multi-controlled-$Z$ gate.
The $K$-qubit TACU gadget can be used to apply a multi-qubit phase gate as
\begin{align}
  \vcenter{\hbox{\includegraphics{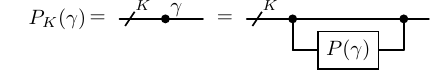}}}
  \label{eq:tacu-phase}
\end{align}
where $P(\gamma) = P_1(\gamma) = e^{-i\gamma\op{1}}$.

We say that a circuit $V$ is $Z$-equivalent to the unitary $U$ if the action of $U$ is equal to the action of $V$ followed by single-qubit Pauli-$Z$ corrections that can be efficiently computed from the outcomes of measurements performed in $V$.
For the purposes of the work below, we formally define a \textit{logical cycle} to be the time required to perform one logical two-qubit Pauli operator measurement.
In practice, one logical cycle in a distance-$d$ surface code on a two-dimensional architecture with local interactions is the time required to perform $d$ rounds of syndrome measurement \cite{litinski2019game}.

The main technical result that we wish to establish in this section is the following:
\begin{theorem}
  The $K$-qubit phase gate $P_K(\gamma) = e^{i\gamma\op{1}^{\otimes K}}$ is $Z$-equivalent to a circuit that addresses the $K$ qubits of $P_K(\gamma)$ for one logical cycle, introduces $\lfloor\frac{13}{2} K\rfloor$ ancilla qubits, consumes $\sum_{\ell=1}^{\lceil\log_2 K\rceil} \lfloor K/2^\ell\rfloor \le K-1$ CCZ states, and can be implemented in a total of $n_P + 4\lceil\log_2 K\rceil$ logical cycles, where $n_P$ is the number of logical cycles required to implement the single-qubit phase gate $P_1(\gamma)$.
  \label{thm:phase-gate-appendix}
\end{theorem}
Here an \textit{adaptive measurement} is a measurement whose basis may depend on (and can be efficiently computed from) the outcomes of past measurements.
We prove Theorem \ref{thm:phase-gate-appendix} in stages below, by first optimizing the two-qubit TACU gadget in Ref.~\cite{litinski2022active}, and then combining two-qubit TACU gadgets into a $K$-qubit TACU gadget.

\subsection{Optimizing the two-qubit TACU gadget}

Starting with the two-qubit TACU gadget in Fig.~15(c) of Ref.~\cite{litinski2022active}, we can expand the conditional (reactive) CZ gates therein using Fig.~14(b) of Ref.~\cite{litinski2022active} to write
\begin{align}
  \vcenter{\hbox{\includegraphics{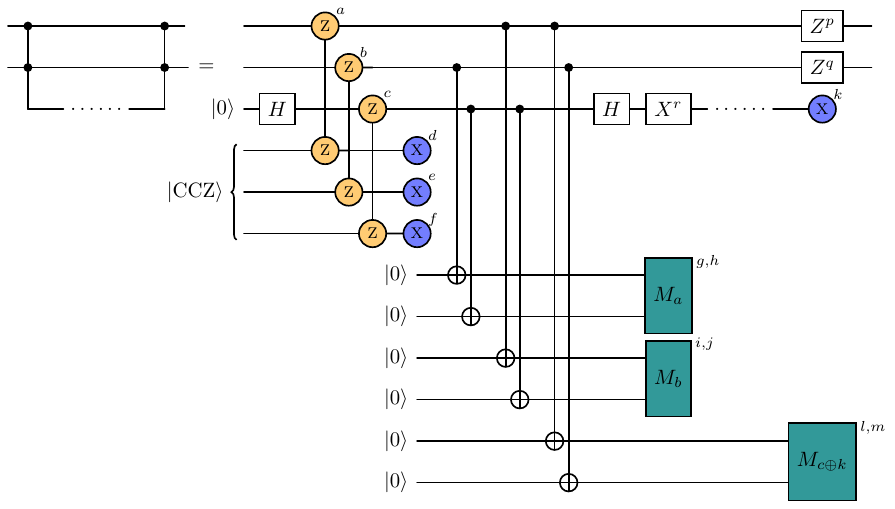}}}.
  \label{eq:tacu-start}
\end{align}
Here circles represent single- and two-qubit Pauli measurements in a basis that is indicated by both text and color, for clarity, and whose outcomes are saved to bits $a,b,c,d,e,f,k$; we define the conditional two-qubit measurement gate
\begin{align}
  \vcenter{\hbox{\includegraphics{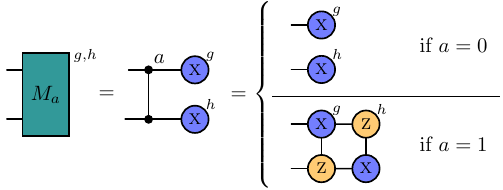}}},
  \label{eq:measure-pair}
\end{align}
where the $a$ superscript on the CZ gate indicates that this gate is applied only if $a=1$; and the control bits $p,q,r$ for the Pauli gates in Eq.~\eqref{eq:tacu-start} are
\begin{align}
  p &= (b\wedge c)\oplus d\oplus i\oplus l, \\
  q &= (a\wedge c)\oplus e\oplus g\oplus m, \\
  r &= (a\wedge b)\oplus f\oplus h\oplus j.
  \label{eq:pauli-controls}
\end{align}
We can ``offload'' gates from the top qubits in Eq.~\eqref{eq:tacu-start} by applying the circuit identity
\begin{align}
  \vcenter{\hbox{\includegraphics{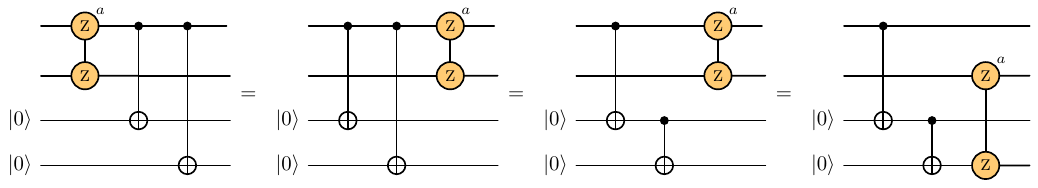}}},
  \label{eq:offload}
\end{align}
where the last equality holds because the joint state of the top qubit and the bottom two ancillas after the two CNOTs looks like $\psi_0\ket{000} + \psi_1\ket{111}$, so the two-qubit $ZZ$ measurement may equivalently address any of these qubits.
We thereby find that
\begin{align}
  \vcenter{\hbox{\includegraphics{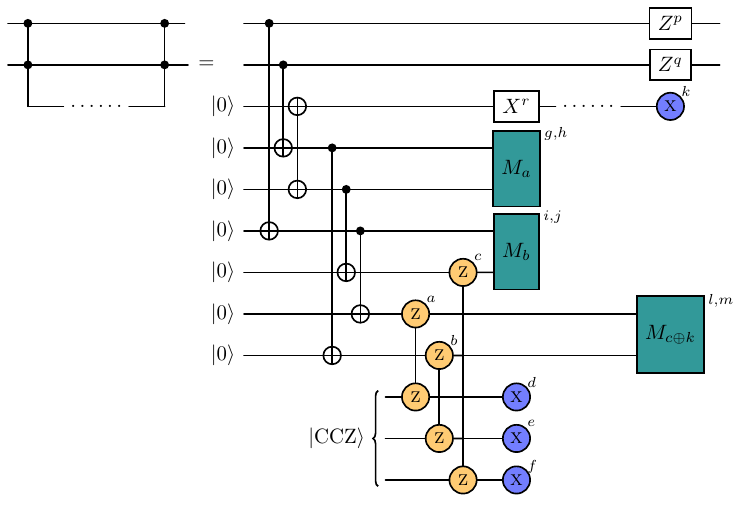}}},
  \label{eq:tacu-mid}
\end{align}
where we use the $X$-controlled-$X$ gate on the third and fifth qubits, which is simply a CZ gate in the $X$ basis.
Finally, we use a decomposition of the CNOT gate into lattice surgery primitives \cite{vuillot2019code},
\begin{align}
  \vcenter{\hbox{\includegraphics{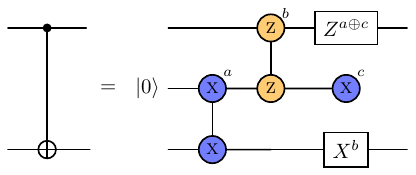}}},
\end{align}
to decompose each chain of interactions between the top and bottom three qubits in Eq.~\eqref{eq:tacu-mid} as
\begin{align}
  \vcenter{\hbox{\includegraphics{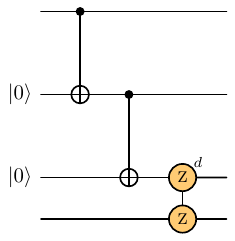}}}
  &= \vcenter{\hbox{\includegraphics{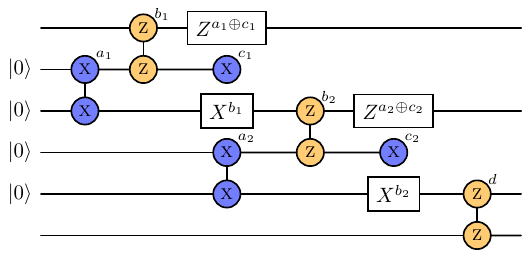}}}
  \label{eq:compress-1} \\
  &= \vcenter{\hbox{\includegraphics{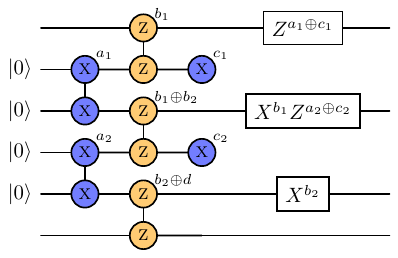}}},
  \label{eq:compress-2}
\end{align}
where the superscripts $b_1\oplus b_2$ and $b_2\oplus d$ on $ZZ$ measurements indicate that these measurement outcomes are equated with $b_1\oplus b_2$ and $b_2\oplus d$, respectively, which implicitly determines the values of $b_2$ and $d$.
We thus find that the two-qubit TACU gadget can be implemented by a circuit of the form
\begin{align}
  \vcenter{\hbox{\includegraphics{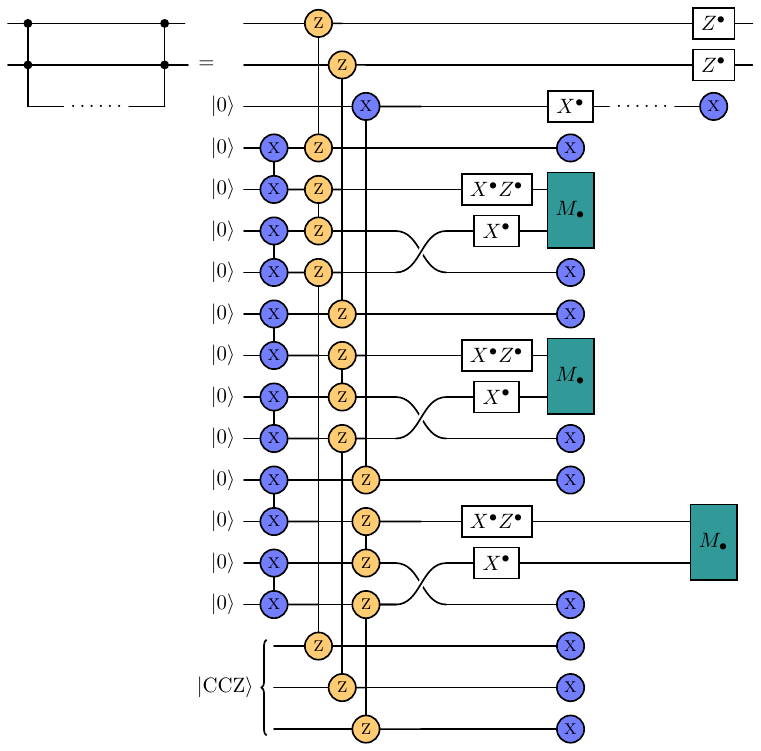}}},
  \label{eq:tacu-final}
\end{align}
where bullets $(\bullet)$ indicate unspecified, but simple dependencies on measurement outcomes in the circuit.
In total, the two-qubit TACU gadget addresses its target qubits for one logial cycle, consumes one CCZ state with 13 $\ket{0}$-state ancilla qubits, requires three logical cycles to ``write'' the AND of its target qubits to an ancilla qubit for further processing, and ends with a measurement that requires one logical cycle to determine adaptive Pauli-$Z$ corrections to the target qubits.

\subsection{The multi-qubit TACU gadget}

The circuit in Eq.~\eqref{eq:tacu-final} can be schematically written as
\begin{align}
  \vcenter{\hbox{\includegraphics{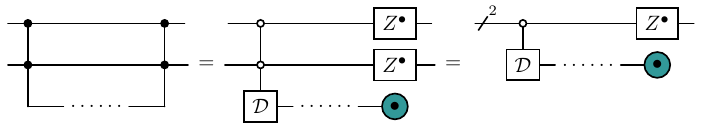}}},
  \label{eq:tacu-dispatch}
\end{align}
where $\mathcal{D}$ is an appropriately defined \textit{TACU dispatch circuit}, and the final measurement in Eq.~\eqref{eq:tacu-dispatch} includes both of the rightmost measurement gates in Eq.~\eqref{eq:tacu-final}.
This notation thereby suppresses the presence of the two ancilla qubits introduced in $\mathcal{D}$ that are measured at the end of Eq.~\eqref{eq:tacu-dispatch}; namely, the ancilla qubits that participate in the last conditional measurement ($M_\bullet$) gate in Eq.~\eqref{eq:tacu-final}.

If $K=2^L$ for integer $L$, then we can construct a $K$-qubit TACU gadget by AND-ing together pairs of qubits in a binary tree.
Defining $K_\ell = K/2^\ell$ and denoting $j$ concurrent copies of the two-qubit TACU dispatch circuit by $\mathcal{D}_j$, we thus find that
\begin{align}
  \vcenter{\hbox{\includegraphics{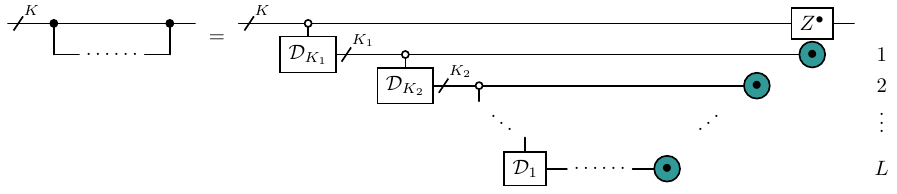}}},
  \label{eq:tacu-k-dispatch}
\end{align}
where the final measurements have to occur sequentially from bottom to top because the measurements at level $\ell$ may induce Pauli-$Z$ corrections at level $\ell-1$ that affect the measurement basis.

If $K$ is not a power of two, we instead define $K_\ell = \lceil K/2^\ell\rceil$ and $L=\lceil\log_2 K\rceil$, and replace two-qubit TACU gadgets in Eq.~\eqref{eq:tacu-k-dispatch} as needed to address unpaired qubits by the ``one-qubit TACU'' gadget
\begin{align}
  \vcenter{\hbox{\includegraphics{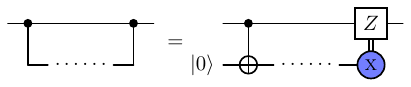}}},
  \label{eq:one-tacu}
\end{align}
which can be thought of as a temporary-COPY-compute-and-uncompute gadget.
In total, the $K$-qubit TACU gadget consumes $\sum_{\ell=1}^{\lceil\log_2 K\rceil} K_\ell \le K - 1$ CCZ states, requires $3\lceil\log_2 K\rceil$ logical cycles to write the AND of its target qubits to an ancilla qubit for further processing, and ends with $\lceil\log_2 K\rceil$ logical cycles to measure out the ancillas.
The dispach circuits $\mathcal{\mathcal{D}}_{K_1}$ introduce at most $\lfloor\frac{13}{2} K\rfloor$ ancilla qubits, and the majority of these ancillas become free at the time of dispatching $\mathcal{\mathcal{D}}_{K_2}$, which in turn needs half the number of ancillas as $\mathcal{\mathcal{D}}_{K_1}$.
In total, $\lfloor\frac{13}{2} K\rfloor$ ancilla qubits thereby suffice for the $K$-qubit TACU gadget.
These resource counts, together with the facts that (a) when ignoring Pauli-$Z$ corrections, data qubits are addressed only once by the $K$-qubit TACU gadget in Eq.~\eqref{eq:tacu-k-dispatch}, and (b) the $K$-qubit phase gate $P_K(\gamma)$ can be implemented by inserting of a single-qubit phase gate $P_1(\gamma)$ into Eq.~\eqref{eq:tacu-k-dispatch}, prove Theorem \ref{thm:phase-gate-appendix}.

\section{Resource bounds for the $k$-SAT oracle}
\label{sec:oracle-bounds}

Here we bound the resources required to implement the $k$-SAT oracle $O_C$.
In principle, this oracle can be implemented using the TACU gadgets in \cref{sec:TACU}.
However, a straightforward TACU-based implementation has an ancilla qubit overhead that would dominate the qubit requirements of the QAOA+AA algorithm in the main text.
We therefore construct an implementation of the oracle with reduced space overheads, based on explicit (as opposed to temporary) computation and uncomputation.

\subsection{The multi-qubit Toffoli gate}

Starting with the Toffoli gate in Fig.~14(a) of Ref.~\cite{litinski2022active},
\begin{align}
  \vcenter{\hbox{\includegraphics{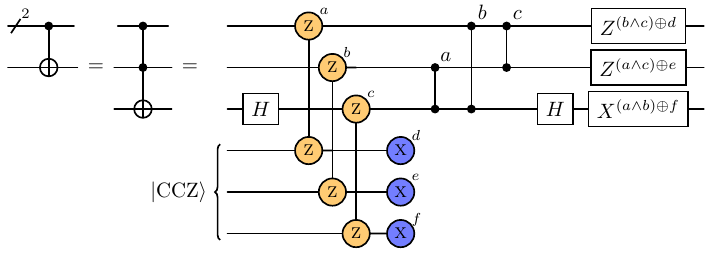}}},
  \label{eq:toffoli}
\end{align}
the three conditional CZ gates can be performed in two logical cycles,
\begin{align}
  \vcenter{\hbox{\includegraphics{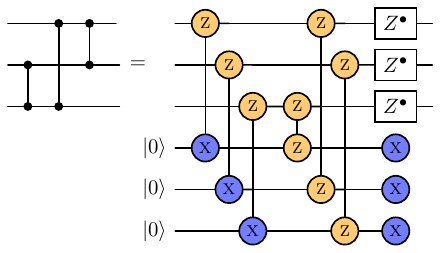}}}.
\end{align}
By commuting through and merging the Hadamard gates in Eq.~\eqref{eq:toffoli}, we can thus implement the Toffoli gate by consuming one CCZ state in three logical cycles using three ancilla qubits.

If $K=2^L$ for integer $L$, then the $K$-qubit TACU gadget in Eq.~\eqref{eq:tacu-k-dispatch} and the three-qubit Toffoli gate in Eq.~\eqref{eq:toffoli} can be combined to construct a $(K+1)$-qubit Toffoli gate
\begin{align}
  \vcenter{\hbox{\includegraphics{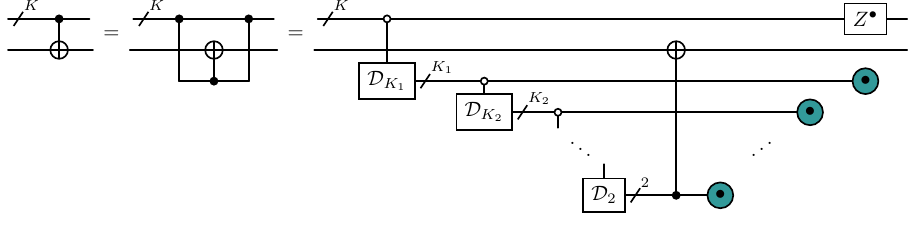}}},
  \label{eq:toffoli-k}
\end{align}
where $K_\ell = \lfloor K/2^\ell\rfloor$, and we make a minor optimization that eliminates the final dispatch ($\mathcal{D}_1$) in Eq.~\eqref{eq:tacu-k-dispatch} and simplifies the resource accounting below.
If $K$ is not a power of two, then two-qubit TACU gadgets and the Toffoli gate are, respectively, replaced by the one-qubit TACU gadget in Eq.~\eqref{eq:one-tacu} and a CNOT gate as necessary to address unpaired qubits.
This $(K+1)$-qubit Toffoli gate thereby consumes up to $\sum_{\ell=1}^{\lceil\log_2 K\rceil} K_\ell \le K - 1$ CCZ states, requires $3\lceil\log_2 K\rceil$ logical cycles to finish addressing the target qubit, and ends with $\lceil\log_2 K\rceil - 1$ logical cycles to uncompute all intermediate data on $\lfloor\frac{13}{2} K\rfloor$ ancilla qubits.

\subsection{The oracle}
\label{sec:oracle}

Each clause of a $k$-SAT instance is an OR of $k$ (possibly negated) variables, which can be converted into a $k$-fold AND operation by the appropriate insertion of single-bit negations (i.e., using De Morgan's Law).
We can schematically represent a gadget that temporarily writes clause $C_j$ onto one qubit using the De Morgan's Law and the $k$-qubit TACU gadget in Eq.~\eqref{eq:tacu-dispatch} by
\begin{align}
  \vcenter{\hbox{\includegraphics{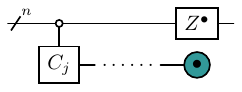}}}.
\end{align}
Let $Q_j$ be the set of qubits that correspond to the variables addressed in the clause $C_j$.
We say that a partition $\mathcal{C} = \set{\mathcal{C}_1,\mathcal{C}_2,\cdots,\mathcal{C}_c}$ of the $k$-SAT clauses $C = \set{C_1,\cdots,C_m}$ is \textit{disjoint} if each part $\mathcal{C}_\ell\subset C$ consists of clauses that address mutually disjoint sets of qubits, which is to say that if $C_i,C_j\in\mathcal{C}_\ell$, then $\abs{Q_i\cap Q_j} = 0$.
Given a disjoint partition $\mathcal{C} = \set{\mathcal{C}_1, \mathcal{C}_2, \cdots, \mathcal{C}_c}$ of $C$ into subsets with maximum size $s = \max_j\abs{\mathcal{C}_j}$, without loss of generality we let $\mathcal{C}_1 = \set{C_1,C_2,\cdots,C_s}$ and define a gadget that writes the AND of all clauses in $\mathcal{C}_1$ onto a single qubit,
\begin{align}
  \vcenter{\hbox{\includegraphics{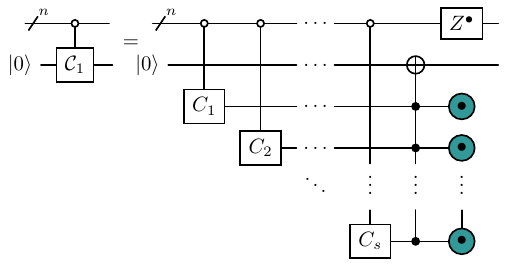}}}.
  \label{eq:one-part}
\end{align}
The gadgets for $\mathcal{C}_2,\mathcal{C}_3,\cdots,\mathcal{C}_c$ are defined analogously.
By construction, all individual clauses in Eq.~\eqref{eq:one-part} can be computed in parallel.
In total, this gadget consumes up to $R_1$ CCZ states in $L_1$ logical cycles using $A_1$ ancilla qubits, where
\begin{align}
  R_1 &= (k-1)s + s - 1 = ks - 1,
  \\
  L_1 &= 4\lceil\log_2k\rceil + 4\lceil\log_2s\rceil - 1
  \approx 4 \log_2(ks),
  \\
  A_1 &= \left\lceil\frac{13}{2} k\right\rceil s + \left\lceil\frac{13}{2}s\right\rceil
  \approx \frac{13}{2}(k+1)s.
\end{align}
We now further partition $\mathcal{C}$ into $\lceil\sqrt{c}\rceil$ subsets $\mathscr{C}_1, \mathscr{C}_2, \cdots, \mathscr{C}_{\lceil\sqrt{c}\rceil}$ of size at most $\lceil\sqrt{c}\rceil$, with for example $\mathscr{C}_1 = \set{\mathcal{C}_1, \mathcal{C}_2, \cdots, \mathcal{C}_{\lceil\sqrt c\rceil}}$.
We can AND together all clauses in $\mathscr{C}_1$ by computing all $\mathcal{C}_j\in\mathscr{C}_1$, AND-ing them together with a $\lceil\sqrt{c}\rceil$-qubit Toffoli, and uncomputing all $\mathcal{C}_j\in\mathscr{C}_1$ with gadgets of the form
\begin{align}
  \vcenter{\hbox{\includegraphics{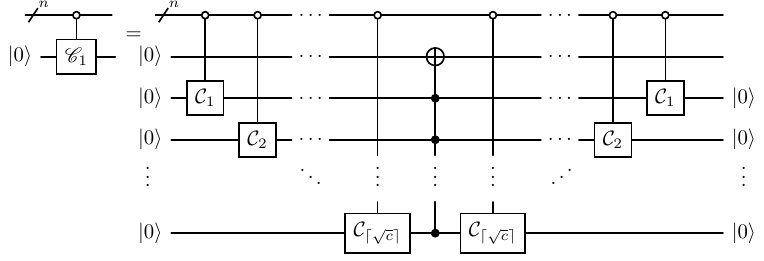}}}.
  \label{eq:many-parts}
\end{align}
If all parts $\mathcal{C}_j$ are computed sequentially, this gadget consumes up to $R_2$ CCZ states in $L_2$ logical cycles using $A_2$ ancilla qubits, where
\begin{align}
  R_2 &= 2\lceil\sqrt{c}\rceil R_1 + \lceil\sqrt{c}\rceil - 1
  \approx 2 ks \sqrt{c},
  \\
  L_2 &= 2\lceil\sqrt{c}\rceil L_1 + 4\log_2\lceil\sqrt{c}\rceil - 1
  \approx 8 \log_2(ks) \sqrt{c},
  \\
  A_2 &= A_1 + \lceil\sqrt{c}\rceil
  \approx \frac{13}{2}(k+1)s + \sqrt{c}.
\end{align}
If $\eta\le L_1$ parts $\mathcal{C}_j$ are computed concurrently, then the number of required logical ancilla qubits grows by a factor of $\eta$, the contribution of $2\lceil\sqrt{c}\rceil L_1$ to $L_2$ gets reduced by a factor of $\eta$, and an additional $\eta-1$ logical cycles are added to $L_2$ to account for the fact that the parts $\mathcal{C}_j$ must be dispatched one logical cycle at a time, such that altogether
\begin{align}
  L_2 = \left\lceil\frac{2\lceil\sqrt{c}\rceil L_1}{\eta}\right\rceil + \eta + 4\log_2\lceil\sqrt{c}\rceil - 2
  \approx 8 \log_2(ks)\sqrt{c} \, \eta^{-1},
  &&
  A_2 = A_1 \eta + \lceil\sqrt{c}\rceil
  \approx \frac{13}{2}(k+1)s\eta + \sqrt{c}.
\end{align}
Finally, the oracle $O_C$ can be implemented by AND-ing all $\mathscr{C}_j$ together onto one qubit, applying a single-qubit phase (Pauli-$Z$) gate, and uncomputing,
\begin{align}
  \vcenter{\hbox{\includegraphics{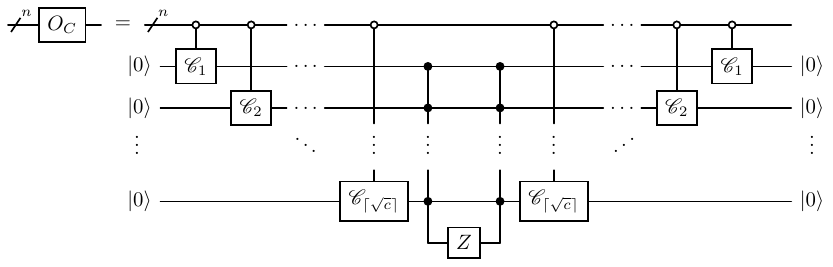}}}.
  \label{eq:oracle}
\end{align}
The oracle $O_C$ therefore consumes up to $R_3$ CCZ states in $L_3$ logical cycles using $A_3$ ancilla qubits, where
\begin{align}
  R_3
  &= 2 \lceil\sqrt{c}\rceil R_2 + \lceil\sqrt{c}\rceil - 1
  = (4ks-2) \lceil\sqrt{c}\rceil^2 - \lceil\sqrt{c}\rceil - 1
  \approx 4ks c, \\
  L_3
  &= 2\lceil\sqrt{c}\rceil L_2 + 4\log_2\lceil\sqrt{c}\rceil
  \approx 16 \log_2(ks) c \eta^{-1}, \\
  A_3
  &= A_2 + \lceil\sqrt{c}\rceil
  \approx \frac{13}{2}(k+1) s \eta + 2\sqrt{c}.
\end{align}
In summary, the oracle $O_C$ can be implemented by a circuit that looks like
\begin{align}
  \vcenter{\hbox{\includegraphics{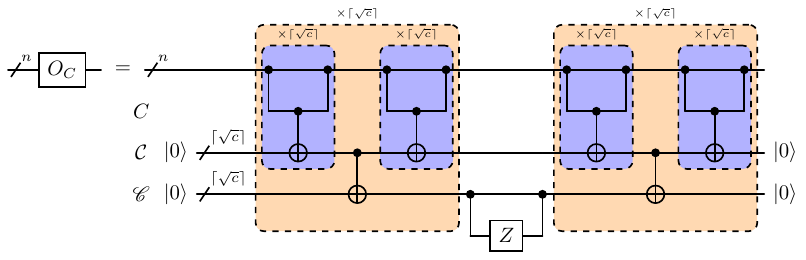}}}.
\end{align}
In the main text, the space overheads of the parallelized phaser exceed those of the unparallelized ($\eta=1$) oracle.
We therefore parallelize the oracle to the point at which its space overheads (from both logical ancilla qubits and CCZ factories) nearly match (but do not exceed) the space overheads of the parallelized phaser.
In practice, this means that $\eta\in[\oracleratiofaconespeeduptwo,\oracleratiofaconespeedupfour]$ for all cases considered in the main text.

\section{Optimal QAOA depths}

As illustrated by the time budget in \cref{fig:techniques}D, the QAOA+AA runtime $T_q$ is dominated by the time to implement the QAOA phaser.
Therefore, to good approximation this runtime goes $ T_q(n,p) \propto p /\sqrt{\Pr^{\rm success}_{\rm QAOA}} = p 2^{0.69p^{-0.32}n/2}$, illustrated also in \cref{fig:fig_optimal_p}.
This dependence on $n$ and $p$ allows us to compute an optimal QAOA depth $p$ for every problem size $n$, namely
\begin{align}
  p_{\mathrm{opt}}(n)
  = \min_p T_q(n,p)
  \approx \left(\frac{\ln 2}{2}\times 0.69 \times 0.32 \times n\right)^{\frac{1}{0.32}}
  \approx 3.25\times10^{-4} \times n^{\frac{1}{0.32}}.
  \label{eq:optimal_p_8_SAT}
\end{align}
We note, however, that this optimization relies on a strong commitment to the functional form $T_q(n,p)$ and its precise parameters ($0.69$ and $0.32$) for arbitrary QAOA depths $p$.
We therefore instead consider, in the main text, specific values of $p$ that correspond to different asymptotic speedups, which has the added benefit of showcasing the plausibility of a quantum advantage in optimization with other quantum algorithms that typically have fixed asymptotic speedups.

\begin{figure}
  \centering
  \includegraphics[width=0.7\linewidth]{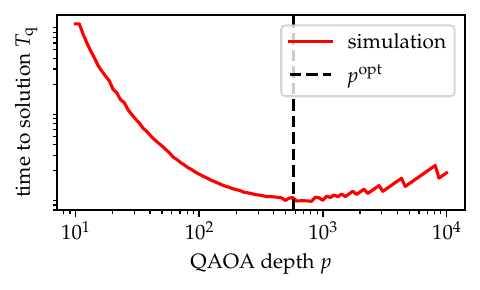}
  \caption{
    Dependence of the time-to-solution $T_q$ for QAOA+AA on the QAOA depth $p$ for $n=100$ 8-SAT variables.
  }
  \label{fig:fig_optimal_p}
\end{figure}

\begin{sidewaystable}
  \centering
  \begin{longtable}{|c|c|c|c|c|c|c|c|c|c|c|c|c|c|c|}
    \hline
    asymptotic
    & QAOA
    & problem
    & code
    & physical
    & classical
    & logical
    & non-Clifford
    & parallel
    & number of
    & number of
    & number of
    & decomposition
    & non-Clifford
    & crossover \\
    quantum
    & cycles
    & size
    & distance
    & qubits
    & decoder
    & depth
    & gates
    & jobs
    & classical
    & logical
    & $T$ gate
    & accuracy
    & infidelity
    & time \\
    speedup
    & $p$
    & $n$
    & $d$ & ($\times 10^6$)
    & $n_{\mathrm{decoder}}$
    & ($\times10^8)$
    & ($\times10^{12})$
    & $n_{\mathrm{jobs}}$
    & cores $n_{\mathrm{cores}}$
    & ancillas
    & for rotation $\mathcal{N}_T$
    & $\delta$
    & $\epsilon_T$
    & $T_q \leq T_c$  \\
    \hline
    & & & & & & & & & &  & & & & \\
    quadratic
    & \pfaconespeeduptwo
    & \crossoverqubitsfaconespeeduptwo
    & \distancefaconespeeduptwo
    & \Tpqubitsfaconespeeduptwo
    & \decodersfaconespeeduptwo k
    & \Tdepthfaconespeeduptwoa
    & \Tcountfaconespeeduptwo
    & \Tfactoriesfaconespeeduptwo
    & \coresfaconespeeduptwo
    & \ancillafaconespeeduptwo
    & \Trotationsfaconespeeduptwo
    & $\rotfaconespeeduptwoa \times 10^{\rotfaconespeeduptwob}$
    & $\Tfaconespeeduptwoa \times 10^{\Tfaconespeeduptwob}$
    & \yearsfaconespeeduptwo{} y
    \\
    & & & & & & & & & & & & & & \\
    \hline
    & & & & & & & & & & & & & & \\
    cubic
    & \pfaconespeedupthree
    & \crossoverqubitsfaconespeedupthree
    & \distancefaconespeedupthree
    & \Tpqubitsfaconespeedupthree
    & \decodersfaconespeedupthree k
    & \Tdepthfaconespeedupthreea
    & \Tcountfaconespeedupthree
    & \Tfactoriesfaconespeedupthree
    & \coresfaconespeedupthree
    & \ancillafaconespeedupthree
    & \Trotationsfaconespeedupthree
    & $\rotfaconespeedupthreea \times 10^{\rotfaconespeedupthreeb}$
    & $\Tfaconespeedupthreea \times 10^{\Tfaconespeedupthreeb}$
    & \hoursfaconespeedupthree{} h
    \\
    & & & & & & & & & & & & & &
    \\ \hline
    & & & & & & & & & & & & & &
    \\
    quartic
    & \pfaconespeedupfour
    & \crossoverqubitsfaconespeedupfour
    & \distancefaconespeedupfour
    & \Tpqubitsfaconespeedupfour
    & \decodersfaconespeedupfour k
    & \Tdepthfaconespeedupfoura
    & \Tcountfaconespeedupfour
    & \Tfactoriesfaconespeedupfour
    & \coresfaconespeedupfour
    & \ancillafaconespeedupfour
    & \Trotationsfaconespeedupfour
    & $\rotfaconespeedupfoura \times 10^{\rotfaconespeedupfourb}$
    & $\Tfaconespeedupfoura \times 10^{\Tfaconespeedupfourb}$
    & \hoursfaconespeedupfour{} h
    \\
    & & & & & & & & & & & & & &\\ \hline
  \end{longtable}
  \caption{\textbf{Extended data for the crossover point.}
  This extended table shows parameters, beyond those in Fig~\ref{fig:main_results} of the main text, for crossover points at which QAOA+AA and the classical solver Sparrow take equal time to solve, in expectation, random instances of 8-SAT near the satisfiability threshold.}
\end{sidewaystable}

\begin{sidewaystable}
  \centering
  \begin{longtable}{|c|c|c|c|c|c|c|c|c|c|c|c|c|c|c|}
    \hline
    asymptotic
    & QAOA
    & problem
    & code
    & physical
    & classical
    & logical
    & non-Clifford
    & parallel
    & number of
    & number of
    & number of
    & decomposition
    & non-Clifford
    & crossover \\
    quantum
    & cycles
    & size
    & distance
    & qubits
    & decoder
    & depth
    & gates
    & jobs
    & classical
    & logical
    & $T$ gate
    & accuracy
    & infidelity
    & time \\
    speedup
    & $p$
    & $n$
    & $d$ & ($\times 10^6$)
    & $n_{\mathrm{decoder}}$
    & ($\times10^8)$
    & ($\times10^{12})$
    & $n_{\mathrm{jobs}}$
    & cores $n_{\mathrm{cores}}$
    & ancillas
    & for rotation $\mathcal{N}_T$
    & $\delta$
    & $\epsilon_T$
    & $T_q \leq T_c$  \\
    \hline
    & & & & & & & & & &  & & & & \\
    quadratic
    & \pfaconespeeduptwo
    & \improvcrossoverqubitsfactwospeeduptwo
    & \improvdistancefactwospeeduptwo
    & \improvTpqubitsfactwospeeduptwo
    & \improvdecodersfactwospeeduptwo k
    & \improvTdepthfactwospeeduptwoa
    & \improvTcountfactwospeeduptwo
    & \improvTfactoriesfactwospeeduptwo
    & \improvcoresfactwospeeduptwo
    & \improvancillafactwospeeduptwo
    & \improvTrotationsfactwospeeduptwo
    & $\improvrotfactwospeeduptwoa \times 10^{\improvrotfactwospeeduptwob}$
    & $\improvTfactwospeeduptwoa \times 10^{\improvTfactwospeeduptwob}$
    & \improvdaysfactwospeeduptwo{} d
    \\
    & & & & & & & & & & & & & & \\
    \hline
    & & & & & & & & & & & & & & \\
    cubic
    & \pfaconespeedupthree
    & \improvcrossoverqubitsfactwospeedupthree
    & \improvdistancefactwospeedupthree
    & \improvTpqubitsfactwospeedupthree
    & \improvdecodersfactwospeedupthree k
    & \improvTdepthfactwospeedupthreea
    & \improvTcountfactwospeedupthree
    & \improvTfactoriesfactwospeedupthree
    & \improvcoresfactwospeedupthree
    & \improvancillafactwospeedupthree
    & \improvTrotationsfactwospeedupthree
    & $\improvrotfaconespeedupthreea \times 10^{\improvrotfactwospeedupthreeb}$
    & $\improvTfactwospeedupthreea \times 10^{\improvTfactwospeedupthreeb}$
    & \improvhoursfactwospeedupthree{} h
    \\
    & & & & & & & & & & & & & &
    \\ \hline
    & & & & & & & & & & & & & &
    \\
    quartic
    & \pfaconespeedupfour
    & \improvcrossoverqubitsfactwospeedupfour
    & \improvdistancefactwospeedupfour
    & \improvTpqubitsfactwospeedupfour
    & \improvdecodersfactwospeedupfour k
    & \improvTdepthfactwospeedupfoura
    & \improvTcountfactwospeedupfour
    & \improvTfactoriesfactwospeedupfour
    & \improvcoresfactwospeedupfour
    & \improvancillafactwospeedupfour
    & \improvTrotationsfactwospeedupfour
    & $\improvrotfactwospeedupfoura \times 10^{\improvrotfactwospeedupfourb}$
    & $\improvTfactwospeedupfoura \times 10^{\improvTfactwospeedupfourb}$
    & \improvhoursfactwospeedupfour{} h
    \\
    & & & & & & & & & & & & & &\\ \hline
  \end{longtable}
  \caption{\textbf{Extended data for the crossover point for the combined improvements  for realistic
classical
parallelization.}
  This extended table shows parameters, beyond those in \cref{tab:improvements} of the main text, for crossover points at which QAOA+AA and the classical solver Sparrow take equal time to solve, in expectation, random instances of 8-SAT near the satisfiability threshold. Classical solver is run on all $725,760$ cores of MareNostrum 5 GPP supercomputer~\cite{top500} using the impact of parallelization estimated in Ref.~\cite{arbelaez2013using}.}
\end{sidewaystable}

\begin{sidewaystable}
  \centering
  \begin{longtable}{|c|c|c|c|c|c|c|c|c|c|c|c|c|c|c|}
    \hline
    asymptotic
    & QAOA
    & problem
    & code
    & physical
    & classical
    & logical
    & non-Clifford
    & parallel
    & number of
    & number of
    & number of
    & decomposition
    & non-Clifford
    & crossover \\
    quantum
    & cycles
    & size
    & distance
    & qubits
    & decoder
    & depth
    & gates
    & jobs
    & classical
    & logical
    & $T$ gate
    & accuracy
    & infidelity
    & time \\
    speedup
    & $p$
    & $n$
    & $d$ & ($\times 10^6$)
    & $n_{\mathrm{decoder}}$
    & ($\times10^8)$
    & ($\times10^{12})$
    & $n_{\mathrm{jobs}}$
    & cores $n_{\mathrm{cores}}$
    & ancillas
    & for rotation $\mathcal{N}_T$
    & $\delta$
    & $\epsilon_T$
    & $T_q \leq T_c$  \\
    \hline
    & & & & & & & & & &  & & & & \\
    quadratic
    & \pfaconespeeduptwo
    & \pessimisticcrossoverqubitsfaconespeeduptwo
    & \pessimisticdistancefaconespeeduptwo
    & \pessimisticTpqubitsfaconespeeduptwo
    & \pessimisticdecodersfaconespeeduptwo k
    & \pessimisticTdepthfaconespeeduptwoa
    & \pessimisticTcountfaconespeeduptwo
    & \pessimisticTfactoriesfaconespeeduptwo
    & \pessimisticcoresfaconespeeduptwo
    & \pessimisticancillafaconespeeduptwo
    & \pessimisticTrotationsfaconespeeduptwo
    & $\pessimisticrotfaconespeeduptwoa \times 10^{\pessimisticrotfaconespeeduptwob}$
    & $\pessimisticTfaconespeeduptwoa \times 10^{\pessimisticTfaconespeeduptwob}$
    & \pessimisticyearsfaconespeeduptwo{} y
    \\
    & & & & & & & & & & & & & & \\
    \hline
    & & & & & & & & & & & & & & \\
    cubic
    & \pfaconespeedupthree
    & \pessimisticcrossoverqubitsfaconespeedupthree
    & \pessimisticdistancefaconespeedupthree
    & \pessimisticTpqubitsfaconespeedupthree
    & \pessimisticdecodersfaconespeedupthree k
    & \pessimisticTdepthfaconespeedupthreea
    & \pessimisticTcountfaconespeedupthree
    & \pessimisticTfactoriesfaconespeedupthree
    & \pessimisticcoresfaconespeedupthree
    & \pessimisticancillafaconespeedupthree
    & \pessimisticTrotationsfaconespeedupthree
    & $\pessimisticrotfaconespeedupthreea \times 10^{\pessimisticrotfactwospeedupthreeb}$
    & $\pessimisticTfaconespeedupthreea \times 10^{\pessimisticTfaconespeedupthreeb}$
    & \pessimistichoursfaconespeedupthree{} h
    \\
    & & & & & & & & & & & & & &
    \\ \hline
    & & & & & & & & & & & & & &
    \\
    quartic
    & \pfaconespeedupfour
    & \pessimisticcrossoverqubitsfaconespeedupfour
    & \pessimisticdistancefaconespeedupfour
    & \pessimisticTpqubitsfaconespeedupfour
    & \pessimisticdecodersfaconespeedupfour k
    & \pessimisticTdepthfaconespeedupfoura
    & \pessimisticTcountfaconespeedupfour
    & \pessimisticTfactoriesfaconespeedupfour
    & \pessimisticcoresfaconespeedupfour
    & \pessimisticancillafaconespeedupfour
    & \pessimisticTrotationsfaconespeedupfour
    & $\pessimisticrotfaconespeedupfoura \times 10^{\pessimisticrotfaconespeedupfourb}$
    & $\pessimisticTfaconespeedupfoura \times 10^{\pessimisticTfaconespeedupfourb}$
    & \pessimistichoursfaconespeedupfour{} h
    \\
    & & & & & & & & & & & & & &\\ \hline
  \end{longtable}
  \caption{\textbf{Extended data for the crossover point for the combined improvements  for perfect
classical
parallelization.}
  This extended table shows parameters, beyond those in \cref{tab:improvements} of the main text, for crossover points at which QAOA+AA and the classical solver Sparrow take equal time to solve, in expectation, random instances of 8-SAT near the satisfiability threshold. Classical solver is run on all $725,760$ cores of MareNostrum 5 GPP supercomputer~\cite{top500} with perfect parallelization assumed (i.e., the runtime of classical solver is the serial runtime divided by $n_{\mathrm{cores}}$).}
\end{sidewaystable}

\end{document}